\documentclass[letterpaper,fleqn,11pt]{article}
\usepackage[margin=1in]{geometry}
\usepackage{amsmath,amssymb,amsthm}
\usepackage{mathtools}
\usepackage{thm-restate}
\usepackage{array}
\usepackage{enumitem}
\usepackage{environ}
\usepackage{tikz}
\usetikzlibrary{decorations.pathreplacing,arrows.meta,fit, shapes.geometric}
\usepackage{xcolor}
\usepackage{comment}
\usepackage[bookmarksnumbered=true,hypertexnames=false]{hyperref}
\usepackage[noend]{algpseudocode}
\usepackage{algorithm}
\usepackage[capitalize,sort]{cleveref}
\usepackage{todonotes}

\newenvironment*{tightenum}{\enumerate[noitemsep]}{\endenumerate}

\newcommand*{\TCNLA}[1][]{}  
\newcommand*{\ifNoAppendix}[1]{}

\def\colorschemesepia{sepia}
\def\colorschemedark{dark}
\def\colorschemelight{light}

\ifx\colorscheme\undefined
\let\colorscheme\colorschemelight
\fi

\ifx\colorscheme\colorschemelight
\colorlet{textColor}{black}
\colorlet{bgColor}{white}
\fi

\ifx\colorscheme\colorschemesepia
\definecolor{textColor}{HTML}{433423}
\definecolor{bgColor}{HTML}{fbf0da}
\fi

\ifx\colorscheme\colorschemedark
\definecolor{textColor}{HTML}{bdc1c6}
\definecolor{bgColor}{HTML}{202124}
\definecolor{textBlue}{HTML}{8ab4f8}
\definecolor{textRed}{HTML}{f9968b}
\definecolor{textGreen}{HTML}{81e681}
\definecolor{textPurple}{HTML}{c58af9}
\else
\colorlet{textBlue}{blue!50!black}
\colorlet{textRed}{red!50!black}
\colorlet{textGreen}{green!50!black}
\definecolor{textPurple}{HTML}{681da8}
\fi

\ifx\colorscheme\colorschemelight\else
\pagecolor{bgColor}
\color{textColor}
\fi
\allowdisplaybreaks

\newcommand*{\defeq}{:=}

\newcommand*{\I}{\mathcal{I}}
\newcommand*{\Ical}{\mathcal{I}}
\newcommand*{\Icalhat}{\widehat{\mathcal{I}}}

\newcommand*{\mms}{\text{MMS}}
\newcommand*{\MMS}{\mathrm{MMS}}

\DeclareMathOperator{\approxMMS}{\mathtt{approxMMS}}

\DeclareMathOperator{\toOrdered}{\mathtt{toOrd}}

\DeclareMathOperator{\normalize}{\mathtt{normalize}}
\DeclareMathOperator{\normalizeHyp}{\hyperref[algo:normalize]{\normalize}}
\DeclareMathOperator{\reduce}{\mathtt{reduce}}
\DeclareMathOperator{\reduceHyp}{\hyperref[sec:valid-redn]{\reduce}}
\DeclareMathOperator{\bagFill}{\mathtt{bagFill}}

\newtheorem{lemma}{Lemma}
\newtheorem{theorem}{Theorem}

\newtheorem{corollary}{Corollary}
\newtheorem{definition}{Definition}

\crefname{claim}{Claim}{Claims}
\newtheorem{remark}{Remark}
\theoremstyle{definition}
\newtheorem{example}[theorem]{Example}
\algnewcommand{\LineComment}[1]{\State \textcolor{gray}{\texttt{//} \textit{#1}}}

\NewEnviron{savedProof}[1]{%
\expandafter\expandafter\expandafter\gdef\expandafter\expandafter\csname sp:#1\endcsname\expandafter{%
\expandafter\proof\BODY\endproof}
\proof\BODY\endproof}
\newcommand*{\recallProof}[1]{\csname sp:#1\endcsname}

\newcommand*{\callMacro}[1]{\csname#1\endcsname}

\hypersetup{colorlinks,linkcolor=textRed,citecolor=textRed,urlcolor=textBlue}

\title{Improved MMS Approximations for Few Agent Types}
\author{Jugal Garg\footnote{University of Illinois, Urbana-Champaign. Supported by NSF Grants CCF-1942321 and CCF-2334461}\\\texttt{jugal@illinois.edu}
\and Parnian Shahkar\footnote{University of California, Irvine. Supported by NSF Grant CCF-2230414}\\\texttt{shahkarp@uci.edu}
}

\date{\empty}
\begin{document}
\maketitle
\begin{abstract}
We study fair division of indivisible goods under the \emph{maximin share} (MMS) fairness criterion in settings where agents are grouped into a small number of \emph{types}, with agents within each type having identical valuations. For the special case of a single type, an exact MMS allocation is always guaranteed to exist. However, for two or more distinct agent types, exact MMS allocations do not always exist, shifting the focus to establishing the existence of approximate-MMS allocations. A series of works over the last decade has resulted in the best-known approximation guarantee of $\frac{3}{4} + \frac{3}{3836}$. 

In this paper, we improve the approximation guarantees for settings where agents are grouped into two or three types, a scenario that arises in many practical settings. Specifically, we present novel algorithms that guarantee a $\frac{4}{5}$-MMS allocation for two agent types and a $\frac{16}{21}$-MMS allocation for three agent types. Our approach leverages the MMS partition of the majority type and adapts it to provide improved fairness guarantees for all types. 

\end{abstract}

\section{Introduction}
Fair division of resources is a fundamental problem in various multi-agent settings. In this work, we focus on the discrete setting, where a set of indivisible goods needs to be partitioned among agents with additive preferences. The maximin share (MMS) is a widely studied fairness notion in this context. An allocation is said to satisfy MMS fairness if every agent receives a bundle of goods that they value at least their MMS value. The MMS value for an agent represents the maximum value they can guaranteee for themselves by partitioning all goods into $n$ bundles (one for each agent) and receiving the least-valued bundle, where $n$ is the total number of agents. Agents with identical valuation functions are considered the same \emph{type}.

While MMS allocations provide a strong fairness guarantee for all agents, they are not guaranteed to exist for all instances~\cite{procaccia2014fair}. Notably, when all agents have identical valuation functions---i.e., they belong to a single type---an MMS allocation is trivially guaranteed. However, in settings with three or more agents and just two distinct types, MMS allocations may not always exist~\cite{FeigeST21}.  

These non-existence results have shifted the focus toward approximate MMS guarantees. In this framework, an allocation is called $\alpha$-MMS if every agent receives a bundle they value at least $\alpha$ times their MMS value. Over the last decade, extensive research has established the existence of allocations with an approximation guarantee of $\frac34 + \frac{3}{3836}$-MMS; see, e.g.,~\cite{barman2020approximation,ghodsi2018fair,garg2019approximating,Nikzad,10.1145/3140756,garg2021improved,akrami2023breaking34barrierapproximate}. This remains the best-known guarantee, even in settings with agents grouped into two distinct types. 

Improving MMS approximations has proven to be a significant challenge, with progress being relatively slow. It remains unclear how to substantially improve the approximation ratio beyond $\frac34$ for all instances. The existence of $\frac34$-MMS was first established in~\cite{ghodsi2018fair}, and despite extensive efforts by multiple researchers, the ratio has only been slightly improved over the past eight years. This naturally leads to the question: \emph{Can we achieve better approximation ratios for intermediate cases where the number of agent types lies between 1 and $n$?} 

In this paper, we address this question affirmatively by providing novel algorithms that achieve a \( \frac{4}{5} \)-MMS allocation for two agent types, and a \( \frac{16}{21} \)-MMS allocation for three agent types. Beyond their theoretical significance, these results hold practical importance, as many real-world resource allocation scenarios involve agents grouped into a few distinct categories based on shared preferences or needs. For example, co-working spaces may group users as freelancers, startups, or larger companies, while land allocation might involve agricultural, residential, or commercial priorities. These cases often involve just two or three types of agents, making it crucial to design allocation mechanisms that exploit this structure for better outcomes. These special cases have also been explored in several prior works, often motivated by additional intriguing applications; see, e.g.,~\cite{GargMQ23, GargMQ24, GhosalPNV23}.

A key innovation in our work lies in leveraging the 1-MMS partition for valuation function of the majority type, the type with the largest number of agents. We adapt this partition to achieve better approximations. This departs significantly from traditional methods that typically begin with variations of bag-filling algorithms. Similar to \cite{akrami2023breaking34barrierapproximate}, we aim to ensure that each allocated bag contains exactly one high-valued item, enabling tighter bounds for subsequent bag-filling steps. However, rather than constructing these bags from scratch, we modify the initial MMS partition of the majority type. 

Specifically, if some bags contain multiple high-valued items while others have none, we swap items to ensure that every bag contains exactly one high-valued item. The details of this process are provided in Section \ref{sec:preprocess}. While this process may compromise the 1-MMS property for the majority type, we ensure that the value of each modified bag remains at least \( \frac{4}{5} \) times the MMS for the majority type. Therefore, any of these bags can be allocated to the agents from majority type, but the other types will determine which bags. Essentially, we prioritize assigning bags that are considered low-valued by the other types to the majority type, thereby maximizing the remaining value for the non-majority types.

When there are three types of agents, the allocation problem becomes more complicated, as the two non-majority types might disagree on which bags should be allocated to the majority type. In such cases, we categorize the bags into four classes based on their value to the non-majority types: valuable to both, valuable to one but not the other, valuable to neither. By comparing the number of bags in these categories with the number of agents in each type, we develop tailored solutions for all possible cases. The most challenging scenario arises when many bags are low-valued for both non-majority types, leaving them undesired by either. In such instances, we aggregate all items from these bags and employ a sophisticated bag-filling algorithm to ensure that all agents receive allocations that meet their MMS requirements. The algorithms that achieve improved approximation guarantees for scenarios with two or three agent types are detailed in Section \ref{sec:4}. 

While our approach is effective for instances with two and three types, it faces scalability challenges as the number of types increases. Specifically, the number of bag categories grows significantly, and the resulting case analysis becomes increasingly intricate. Therefore, this paper focuses exclusively on scenarios involving two or three agent types. Notably, barring the MMS value computation for each type through a PTAS \cite{woeginger1997polynomial}, all our algorithms run in polynomial time. We will provide formal definitions, notations, and preliminaries in \cref{sec:prelim}.

\subsection{Additional Related Work}
Given the intense study of the MMS fairness notion and its special cases and variants, we focus here on closely related work.

Computing the MMS value of an agent is NP-hard. However, a Polynomial Time Approximation Scheme (PTAS) exists for this computation~\cite{woeginger1997polynomial}. As noted earlier, MMS allocations are not guaranteed to exist for more than two agents with two distinct types~\cite{procaccia2014fair,FeigeST21}. This non-existence has motivated the exploration of approximate MMS allocations to ensure their existence. A series of works has established the current best approximation factor of $\frac34+\frac{3}{3836}$ for all instances~\cite{akrami2023breaking34barrierapproximate}.

Several works have examined special cases. For example, when $m \leq n+3$, where $m$ is the total number of goods, an MMS allocation always exists~\cite{amanatidis2017approximation}. This bound was later improved to $m\leq n+5$~\cite{FeigeST21}. For $n=2$, MMS allocations always exist~\cite{bouveret2016characterizing}. For $n=3$, the MMS approximation was improved from $\frac{3}{4}$ \cite{procaccia2014fair} to $\frac{7}{8}$ \cite{amanatidis2017approximation} to $\frac{8}{9}$ \cite{gourves2019maximin}, and then to $\frac{11}{12}$~\cite{feige2022improved}. For $n=4$, $\frac{4}{5}$-MMS allocations exist~\cite{ghodsi2018fair}. For $n \geq 5$, the best known factor is the general $\frac{3}{4} + \frac{3}{3836}$ bound~\cite{akrami2023breaking34barrierapproximate}. For the special case of (non-personalized) bivalued instances, MMS allocations are known to exist~\cite{Feigebivalued}.

\paragraph{Chores} The MMS notion can naturally be defined for the fair division of chores, where items provide negative value. As with goods, MMS allocations for chores do not always exist~\cite{Aziz2017,FeigeST21}. However, substantial research on approximate MMS allocations for chores has yielded significant results. Notable works~\cite{Aziz2017,barman2020approximation,HuangL21,huang2023reduction} have led to the existence of $\frac{13}{11}$-MMS allocations. For three agents, $\frac{19}{18}$-MMS allocations exist~\cite{feige2022improved}, and for factored instances, MMS allocations are guaranteed~\cite{GargHS24}. Additionally, for the special case of personalized bivalued instances, $\frac{15}{13}$-MMS allocations exist~\cite{GargHS24}.

\section{Preliminaries}\label{sec:prelim}
In this paper, we primarily follow the notations used in previous work~\cite{simple} to be consistent with the literature. However, beginning with \cref{def:oni}, we introduce new notations that are specific to our work.

For any positive integer $n$, let $[n] = \{1, 2, \ldots, n\}$, and for two positive integers $i,j$ where $i<j$, let $[i,j] = \{i, i+1, \cdots,j\}$. A fair division instance $\mathcal{I} = (N, M, \mathcal{V})$ consists of a set of agents $N = [n]$, a set of goods $M = [m]$, and a vector of valuation functions $\mathcal{V} = (v_1, v_2, \ldots, v_n)$. Each valuation function $v_i: 2^M \to \mathbb{R}_{\geq 0}$ represents agent $i$'s preference over subsets of goods. We assume additive valuations, so for all $S \subseteq M$, $v_i(S) = \sum_{g \in S} v_i(\{g\})$. For ease of notation, for all $g \in M$, we use $v_i(g)$ or $v_{i,g}$ instead of $v_i(\{g\})$. Likewise, throughout this paper, we use the notation \( v_i(g, g') \) as a shorthand for \( v_i(\{g, g'\}) \).

For a set $S$ of goods and a positive integer $d$, let $\Pi_d(S)$ be the set of all partitions of $S$ into $d$ bundles. The maximin share (MMS) value for a valuation $v$ is defined as:
\[
\MMS_v^d(S) = \max_{P \in \Pi_d(S)} \min_{j=1}^d v(P_j).
\]
When the instance $\mathcal{I} = (N, M, \mathcal{V})$ is clear from the context, we use the notation $\MMS_{v_i}^n(M)$ as $\MMS_i(\mathcal{I})$ or $\MMS_i$. For agent $i$, an MMS partition $P^i = (P^i_1, P^i_2, \ldots, P^i_n)$ satisfies $\MMS_i = \min_{j \in [n]} v_i(P^i_j)$. An allocation $X$ is MMS if $v_i(X_i) \geq \MMS_i$ for all $i \in N$. More generally, for $0 < \alpha \leq 1$, an allocation $X$ is $\alpha$-MMS if $v_i(X_i) \geq \alpha \cdot \MMS_i$ for all $i \in N$.

For a list of items \( H \), let $|H|$ denote the number of items in the list. Given positive integers \( i\in[|H|] \) and \( j\in [|H|] \) with \( i < j \), let \( H[i: j] \) denote the subset of items in \( H \) from the \( i \)-th to the \( j \)-th position, inclusive. Also, let \( H[-1] \) represent the last item in the list.

\begin{definition}[Ordered instance]
An instance $\mathcal{I} = (N, M, \mathcal{V})$ is ordered if there exists a permutation of the goods $(g_1, g_2, \ldots, g_m)$ such that for all agents $i \in N$, $v_i(g_1) \geq v_i(g_2) \geq \ldots \geq v_i(g_m)$. For any fair division instance $\mathcal{I} = ([n], [m], \mathcal{V})$, the transformation $\mathtt{order}(\mathcal{I})$ produces an ordered instance $\mathcal{I}' = ([n], [m], \mathcal{V}')$, where for each $i \in [n]$ and $j \in [m]$, $v'_i(j)$ is the $j$th largest value in the multiset $\{v_i(g) \mid g \in [m]\}$.
\end{definition}

\begin{theorem}[Theorem 2 in \protect\cite{barman2020approximation}]\label{order-preserves}
    Given an instance $\I$ and an $\alpha$-MMS allocation of $\mathtt{order}(\I)$, one can compute an $\alpha$-MMS allocation of $\I$ in polynomial time.
\end{theorem}
\Cref{order-preserves} implies the transformation $\mathtt{order}$ is $\alpha$-\emph{MMS-preserving}. For ordered instances $\mathcal{I}$, we assume without loss of generality that $v_i(1) \geq v_i(2) \geq \ldots \geq v_i(m)$ for all $i \in [n]$.

\begin{definition}[Normalized instance]
An instance $\mathcal{I} = (N, M, \mathcal{V})$ is \emph{normalized} if for all agents $i \in N$ and all bundles $P^i_j$ in an MMS partition of $i$, $v_i(P^i_j) = 1$. For any fair division instance $\mathcal{I}$, the transformation $\mathtt{normalize}(\mathcal{I})$ computes a normalized instance $\mathcal{I}' = (N, M, \mathcal{V}')$ by determining the MMS partition $P^i$ for each agent $i \in N$ and rescaling valuations: for all $j \in [n]$ and $g \in P^i_j$, set $v'_{i,g} = v_{i,g} / v_i(P^i_j)$. 
\end{definition}

\begin{lemma}[Lemma 4 in \protect\cite{simple}]\label{thm:normalize}
Let $\I' = ([n], [m], \mathcal{V}') = \mathtt{normalize}(\I = ([n], [m], \mathcal{V}))$. Then for any allocation $A$,
$v_i(A_i) \ge v'_i(A_i)\cdot\MMS^n_{v_i}$ for all $i \in N$.
\end{lemma}
Note that in a normalized instance, the MMS value for every agent is 1. Consequently, \Cref{thm:normalize} establishes that $\mathtt{normalize}$ is $\alpha$-MMS-preserving. This means that if an allocation \( A \) is an $\alpha$-MMS allocation for the normalized instance \( \mathcal{I}' \), then \( A \) is also an $\alpha$-MMS allocation for the original instance \( \mathcal{I} \).
Note that in a normalized instance, the total value of all items satisfies \( v'_i([m]) = \sum_{j\in[n]} v'_i(P^i_j) = n \) for every agent \( i \in [n] \). Also, for each agent $i$ and for every MMS partition $Q$ of agent $i$, we have $v'_i(Q_j) = 1$ $\forall j \in [n]$.

Given an instance $\mathcal{I}$, a reduction rule $R(\mathcal{I})$ allocates a subset $S \subseteq M$ of goods to an agent $i$ and produces a new instance $\mathcal{I}' = (N \setminus \{i\}, M \setminus S, \mathcal{V})$.

\begin{definition}[Valid reductions]
A reduction rule $R$ is a \emph{valid $\alpha$-reduction} if, for $R(\mathcal{I}) = (N', M', \mathcal{V})$, where $\{i\} = N \setminus N'$ and $S = M \setminus M'$:
\begin{enumerate}
    \item $v_i(S) \geq \alpha \cdot \mms^{|N|}_{v_i}(M)$, and
    \item $\mms^{|N|-1}_{v_j}(M') \geq \mms^{|N|}_{v_j}(M)$ for all $j \in N'$.
\end{enumerate}
\end{definition}

If $R$ is a valid $\alpha$-reduction and an $\alpha$-MMS allocation $A$ exists for $R(\mathcal{I})$, then an $\alpha$-MMS allocation for $\mathcal{I}$ can be constructed by allocating $S$ to $i$ and distributing the remaining goods as in $A$. We now describe three standard transformations, known as \emph{reduction rules}, and demonstrate their validity.

\begin{definition}[Reduction rules]
\label{defn:redn-rules}
Consider an ordered fair division instance $(N, M, v)$, where $M \defeq \{g_1, \ldots, g_{|M|}\}$
and $v_{i,g_1} \ge \ldots \ge v_{i,g_{|M|}}$ for every agent $i$. Define
\begin{tightenum}
\item $S_1 \defeq \{g_1\}$.
\item $S_2 \defeq \{g_{|N|}, g_{|N|+1}\}$ if $|M| \ge |N|+1$, else $S_2 \defeq \emptyset$.
\item $S_3 \defeq \{g_{2|N|-1}, g_{2|N|}, g_{2|N|+1}\}$ if $|M| \ge 2|N|+1$, else $S_3 \defeq \emptyset$.
\end{tightenum}
Reduction rule $R_k(\alpha)$: If $v_i(S_k) \ge \alpha\cdot\MMS_i$ for some agent $i$,
then give $S_k$ to $i$. A fair division instance is called $R_k(\alpha)$-\emph{irreducible} if $R_k(\alpha)$
cannot be applied, i.e., $v_i(S_k) < \alpha\cdot\MMS_i$ for every agent $i$.
An instance is called \emph{totally}-$\alpha$-\emph{irreducible} if it is
$R_k(\alpha)$-irreducible for all $k \in [3]$.
\end{definition}
\begin{definition}
The $\reduce_\alpha$ operation takes an ordered fair division instance as input and iteratively applies the reduction rules $R_1(\alpha)$, $R_2(\alpha)$, and $R_3(\alpha)$ in any order until the instance becomes totally-$\alpha$-irreducible.
\end{definition}

\begin{lemma}[Lemma~3.1 in \protect\cite{garg2021improved}]\label{lem:reduce}
For an ordered instance and for $0\leq \alpha \le 1$,
$R_1(\alpha)$, $R_2(\alpha)$, and $R_3(\alpha)$ are valid $\alpha$-reductions.
\end{lemma}

\begin{lemma}[Lemmas 2 and 3 in \protect\cite{simple}]
\label{lem:vr-upper-bounds}
Let $\Ical \defeq ([n], [m], \mathcal{V})$ be an ordered instance where
$v_{i,1} \ge \ldots \ge v_{i,m}$ for each agent $i$.
If $\I$ is totally-$\alpha$-irreducable, then $m\geq 2n$, and
for each agent $i$ and every good $j > (k-1)n$, we have $v_{i,j} < \alpha\cdot\MMS_i/k$.
\end{lemma}
\begin{lemma}[Lemma 6 in \protect\cite{simple}]
    \label{lem:upper-13}
    Let $([n], [m], \mathcal{V})$ be an ordered and normalized fair division instance.
    For all $k \in [n]$ and agent $i \in [n]$, if $v_i(k) + v_i(2n-k+1) > 1$,
    then $v_i(2n-k+1) \leq 1/3$ and $v_i(k) > 2/3$.
\end{lemma}

\begin{definition}\label{def:oni}
    For a fair division instance $\I$, define $\Icalhat_\alpha \defeq \mathtt{order} (\normalizeHyp(\reduceHyp_{\alpha}(\mathtt{order}(\Ical))))$ as the $\text{ONI}_\alpha$ instance of $\I$.
\end{definition}

\label{sec:simple:oni-reduce}
\begin{restatable}{lemma}{thmonireduce}
\label{thm:oni-reduce}
    Let $\Ical$ be a fair division instance, and $\Icalhat_\alpha$ be the $\text{ONI}_\alpha$ instance of $\I$. $\Icalhat_\alpha$ is ordered, normalized, and totally-$\alpha$-irreducible. Furthermore, the transformation of $\Ical$ to $\Icalhat_\alpha$ is $\alpha$-MMS-preserving,
    i.e., a $\alpha$-MMS allocation of $\Icalhat_\alpha$ can be used to obtain
    a $\alpha$-MMS allocation of $\Ical$.
\end{restatable}
\begin{proof}
This proof closely follows the proof of Lemma 5 in \cite{simple}, with the key distinction that, since we employ only three reduction rules, $R_1(\alpha), R_2(\alpha), R_3(\alpha)$, for the $\mathtt{reduce}_\alpha$ transformation, and by \cref{lem:reduce}, all three reduction rules are valid $\alpha$-reductions, our lemma holds for all $0 \leq \alpha \leq 1$.

Applying $\mathtt{order}$ to $\Ical$ produces $\Ical^{(1)}$, which is ordered. The transformation $\reduce_{\alpha}$ applied to $\Ical^{(1)}$ produces $\Ical^{(2)}$, which is totally-$\alpha$-irreducible and preserves orderedness. Normalizing $\Ical^{(2)}$ yields $\Ical^{(3)}$, and since $\normalize$ does not increase the ratio of a good's value to the MMS value, $\Ical^{(3)}$ is totally-$\alpha$-irreducible. Moreover, $\Icalhat_\alpha = \mathtt{order}(\Ical^{(3)})$ is normalized, as for each agent, $\mathtt{order}$ only changes the identities of the goods, but the (multi-)set of values of the goods remains the same. Thus, $\Icalhat_\alpha$ is ordered, normalized, and totally-$\alpha$-irreducible. Since $\mathtt{order}$, $\reduce_{\alpha}$, and $\normalize$ are $\alpha$-MMS-preserving, their composition is also $\alpha$-MMS-preserving. 

The sequence of operations is important, particularly the need to apply $\mathtt{order}$ twice. This is because $\reduce$ requires an ordered input, but it may not preserve normalization, while $\normalize$ ensures normalization but may disrupt orderedness.
\end{proof}

\begin{definition}
    Let $\Icalhat_\alpha \defeq ([n], [m], \mathcal{V})$ be an ONI$_\alpha$ instanace, where $v_{i,1} \geq v_{i,2} \geq \ldots \geq v_{i,m}$ for all $i \in [n]$. The items are categorized into high-valued, middle-valued, and low-valued items as follows:
\begin{tightenum}
    \item High-valued (HV) items: $HV = [n]$.
    \item Middle-valued (MV) items: $MV = [n+1, 2n]$.
    \item Low-valued (LV) items: If $m = 2n$\footnote{By Lemma~\ref{lem:vr-upper-bounds}, $m \geq 2n$.}, then $LV = \emptyset$; otherwise, $LV = [2n+1, m]$.
\end{tightenum}
\end{definition}

\begin{corollary}[of Lemma~\ref{lem:vr-upper-bounds}]\label{col:bound}
In an ONI$_\alpha$ instance $\Icalhat_\alpha$, for each agent $i$ and every good $j$: if $j \in HV$, then $v_i(j) < \alpha$; if $j \in MV$, then $v_i(j) < \alpha/2$; and if $j \in LV$, then $v_i(j) < \alpha/3$.
\end{corollary}

\begin{definition}[Single-High-Valued (SHV) Partition]
    Given an ONI$_\alpha$ instance $\Icalhat_\alpha \defeq ([n], [m], \mathcal{V})$, a partition of $[m]$, $A = \{A_1, A_2, \ldots, A_n\}$ is a Single-High-Valued (SHV) partition if each bundle in $A$ contains exactly one HV item.
\end{definition}

\begin{definition}[SHV $\alpha$-MMS Partition]
    Given an ONI$_\alpha$ instance $\Icalhat_\alpha \defeq ([n], [m], \mathcal{V})$, a partition of $[m]$, $A = \{A_1, A_2, \ldots, A_n\}$ is an SHV $\alpha$-$\MMS$ partition if $A$ is an SHV partition and satisfies $v_i(A_i) \geq \alpha$ for all $i \in [n]$.
\end{definition}

\begin{definition}
Agents are classified into types based on their valuation functions. All agents sharing the same valuation function $v: 2^m \rightarrow \mathbb{R}_{\geq 0}$ are considered to be of the same type. 
\end{definition}

\begin{definition}\label{dfn:ktype}
A $k$-type instance \(\I = ([n], [m], \mathcal{V})\) is a fair division instance involving \(k\) distinct agent types, characterized by the set of valuation functions \(V = \{v_1, v_2, \ldots, v_k\}\). Each agent's valuation function belongs to this set, i.e., \(\mathcal{V}_i \in V\) for all \(i \in [n]\). For each \(j \in [k]\), type \(j\) agents are defined as those whose valuation function is \(v_j\).
\end{definition}

\begin{remark}
A $k$-type instance can be fully described by \(\I \defeq ([n], [m], \{T_1, \ldots, T_k\}, \{v_1, \ldots, v_k\})\), where \(\{v_1, v_2, \ldots, v_k\}\) represents the set of valuation functions for the \(k\) types, and \(T_j\) denotes the number of type \(j\) agents for any $j\in[k]$. Since each agent belongs to exactly one of the \(k\) types, \(\sum_{j \in [k]} T_j = n\). We also assume that the type sizes are ordered as \(T_1 \geq T_2 \geq \cdots \geq T_k\).
\end{remark}

\begin{restatable}{claim}{clktype}
\label{cl:ktype}
Let $\I$ be a $k$-type instance. The $\text{ONI}_\alpha$ instance of $\I$, $\Icalhat_\alpha$ is a $k'$-type instance where $k' \leq k$.
\end{restatable}
\begin{proof}
In the instance $\I$, agents of the same type share the same valuation function. Consequently, they experience the same transformations under the operations $\mathtt{order}$, $\mathtt{normalize}$, and $\mathtt{reduce}_\alpha$, and retain the same valuation function in $\Icalhat_\alpha$. However, agents from different types may converge to the same valuation function; for instance, if the valuations of some types are permutations of the same set of $m$ numbers, they become identical after the $\mathtt{order}$ transformation. Consequently, the number of types can only decrease, and $\Icalhat_\alpha$ becomes a $k'$-type instance where $k' \leq k$.
\end{proof}

\begin{definition}
    In an \( \alpha \)-MMS problem, given a $k$-type \( \text{ONI}_\alpha \) instance, a type \( i \) agent is said to \emph{claim} a bundle of items \( B \) if \( v_i(B) \geq \alpha \). Furthermore, a type \( i \) is considered to claim a bundle if any type $i$ agent claims the bundle. 
\end{definition}

In Sections \ref{sec:preprocess} and \ref{sec:two}, where we study the \( \frac{4}{5} \)-MMS problem, type \( i \) claims a bundle \( B \) if \( v_i(B) \geq \frac{4}{5} \). In Section \ref{sec:three}, as we consider the \( \frac{16}{21} \)-MMS problem, type \( i \) claims \( B \) if \( v_i(B) \geq \frac{16}{21} \).

\section{SHV $\frac{4}{5}$-MMS Partition of Same-type Agents}\label{sec:preprocess}

Given a $1$-type \( \text{ONI}_\alpha \) instance, where \( v \) represents the common valuation function shared by all agents, the objective is to find an SHV \( \frac{4}{5} \)-MMS partition. Since all agents belong to the same type, the MMS value is identical for each agent, i.e., \( \MMS_i = \MMS^n_v \) for all \( i \in [n] \). Moreover, having identical valuations implies an MMS partition for any single agent serves as the MMS partition for the entire instance.

The PTAS described in \cite{woeginger1997polynomial} for computing the MMS partition of a single agent can be utilized to obtain a \( (1 - \epsilon) \)-MMS partition in \( \text{poly}(\frac{1}{\epsilon}) \) time.
By choosing $\epsilon = \min(0.04, \frac{1}{5n})$, we can compute a $(1-\epsilon)$-MMS partition in polynomial time. However, such a partition may include bags with multiple HV items, which violates the current problem's constraints. To address this, we design Algorithm~\ref{alg:SHV} to modify the initial partition, ensuring an SHV $\frac{4}{5}$-MMS solution for same-type agents.

\begin{algorithm}[h]
\caption{SHV $\frac{4}{5}$-MMS of same-type agents}\label{alg:SHV}
\begin{algorithmic}[1]
    \State \textbf{Input:} A $1$-type $\text{ONI}_\alpha$ instance $\I = ([n],[m], \{n\}, \{v\})$, $\alpha = \frac{4}{5}$
    \State \textbf{Output:} SHV $\frac{4}{5}-\MMS$. 
    \State $\epsilon \gets \min(0.04, \frac{1}{5n})$.
    \State $A^0 \gets (1-\epsilon)-\MMS$ partition of an agent. 
    \State $A' \gets $ bags in $A^0$ with one HV item.
    \State $A \gets A^0\setminus A'$.
    \While{$|A'| < n$} 
        \State Let $a\in A$ be an arbitrary bag with $k>1$ HV items.
        \State Let $H =[h_1,\cdots, h_k]$ be the list of HV items in $a$ sorted in an ascending order of valuation.
        \State Let $B =\{b_1,\cdots, b_{k-1}\}\subset A$ be $k-1$ bags without any HV items.
        \State $A \gets A \setminus \{a \cup B\}$.
        \State Let $G =\{g_1,\cdots, g_{k-1}\}$ be $k-1$ highest valued items in $\cup_{b\in B}b$. 
        \State Swap each HV item in $H[1:k-1]$ with a unique item in $G$. 
        \State Denote updated bags as $\Bar{a}, \Bar{b}_1, \cdots, \Bar{b}_{k-1}$.
        \If{$v(\Bar{a}) < \frac{4}{5}$}
            \State $P\gets$ items in $\cup^{k-1}_{j=1}\Bar{b}_j$ excluding HV items $H[1:k-1]$.
            \State Construct $k-1$ new empty bags $B' = \{b'_1,\cdots, b'_{k-1}\}$
            \State Put one HV item from $H[1:k-1]$ in each bag of $B'$. 
            \State From $P$ fill each bag in $\Bar{a} \cup B'$ until it is claimed. \Comment{Bag-filling}\label{bagfilling}
            \State Add $\Bar{a}$ and all bags of $B'$ to $A'$.
        \Else  
            \State Add $\Bar{a}$ and any bag in  $\{\Bar{b}_1, \cdots, \Bar{b}_{k-1}\}$ containing one HV to $A'$, and the rest to $A$. \label{line:23}
        \EndIf
    \EndWhile    
    \Return $A'$
\end{algorithmic}
\end{algorithm}

\textbf{Main ideas of Algorithm \ref{alg:SHV}:} 
At a high level, $A$ represents the set of \emph{flawed} bags, while $A'$ consists of \emph{correct} bags. A bag $B$ is deemed \emph{correct} if it contains exactly one HV item and $v(B)\geq\frac{4}{5}$; otherwise, it is considered \emph{flawed}. As long as there exists at least one flawed bag, we iteratively select a subset of flawed bags and ensure that at least one of them is corrected in each iteration.

\noindent First, we prove some invariant properties of the algorithm in the following claim. 
\begin{restatable}{claim}{clgeneral}
\label{cl:general}
At any point in the algorithm, the following invariants hold: 
    \begin{tightenum}
        \item The total number of bags satisfies $|A| + |A'| = n$.
        \item For any $a \in A$, $a$ does not contain exactly one HV item and $v(a) \geq 1 - \epsilon$.
        \item For any $a \in A'$, $a$ contains exactly one HV item and $v(a) \geq \frac{4}{5}$.
        \item If there exists a bag $a \in A$ that contains $k$ HV items, then there must be at least $k-1$ bags in $A$ that do not contain any HV items.
    \end{tightenum}
\end{restatable}

\begin{proof}
We will prove statements in order. 
\begin{enumerate}
    \item Consider any iteration of the while loop where $k$ flawed bags are selected $\{a, b_1, \dots, b_{k-1}\}\subseteq A$. After modifications, one of the following outcomes occurs:
    \begin{itemize}
        \item All $k$ modified bags are corrected and added to $A'$.
        \item Some of the modified bags are corrected and added to $A'$, while the remaining flawed bags are returned to $A$.
    \end{itemize}
    In both cases, the total number of flawed and correct bags remains constant at $n$. This ensures that the invariant $|A| + |A'| = n$ holds throughout the algorithm.

    \item The claim trivially holds for bags that are initially assigned from $A^0$ to $A$, as $A^0$ is a $(1-\epsilon)$-MMS partition. To extend this to bags added to $A$ later in the algorithm, we use induction. 

    A new bag is added to $A$ only in the following scenario: during a particular iteration of the while loop, after swapping, it holds that $v(\Bar{a}) \geq \frac{4}{5}$, and there exists some $b \in \{\Bar{b}_1, \dots, \Bar{b}_{k-1}\}$ such that $b$ does not contain exactly one HV item. 

    Note that during the swapping process, the value of each bag in $B$ can only increase. This is because, for every $l$ items a bag loses, it gains $l$ HV items in return. Formally, $v(\Bar{b}_j) \geq v(b_j)$ for all $j \in [k-1]$. Since $v(b_j) \geq 1 - \epsilon$ by the induction hypothesis, it follows that any newly added bag to $A$ meets the threshold condition.

    \item This claim holds trivially for bags initially added to $A'$ from $A^0$, as $A^0$ is a $(1-\epsilon)$-MMS partition. For bags added to $A'$ during the iterations of the while loop, consider the following cases:
    \begin{enumerate}
        \item Case 1: $v(\Bar{a}) < \frac{4}{5}$ after swapping:  
       In this case, the $k-1$ swapped HV items are placed in separate bags, and all $k-1$ newly constructed bags, along with $\Bar{a}$, are filled with additional items until all meet the value threshold. Once all $k$ bags are validated, they are added to $A'$. Each of these newly added bags contains exactly one HV item and satisfies the value threshold \( v \geq \frac{4}{5} \).
       \item Case 2: $v(\Bar{a}) \geq \frac{4}{5}$ after swapping:  
       Here, $\Bar{a}$ directly meets the value threshold and contains exactly one HV item, so it is added to $A'$. For the other modified bags in $\{\Bar{b}_1, \dots, \Bar{b}_{k-1}\}$, note that the swapping process ensures their value can only increase. Formally, $v(\Bar{b}_j) \geq v(b_j)$ and $v(b_j) \geq 1 - \epsilon$ according to the second statement of this claim. Thus, each bag $\Bar{b}_j$ added to $A'$ also contains exactly one HV item and satisfies the value threshold.
    \end{enumerate}
    Therefore, in all cases, every bag added to $A'$ contains exactly one HV item and has a value of at least $\frac{4}{5}$.

    \item By the first statement of this claim, the total number of bags in $A$ and $A'$ is $n$. Furthermore, by the third statement in this claim, each bag in $A'$ contains exactly one HV item. Since there are a total of $n$ HV items, it follows that the bags in $A$ collectively contain $|A|$ HV items.
    Therefore, if a bag $a \in A$ contains $k$ HV items, there must be at least $k-1$ bags in $A$ that contain no HV items.
\end{enumerate}
\end{proof}

\noindent Now, we prove the main lemma of this section. 
\begin{restatable}{lemma}{lemmainpreprocess}
\label{lem:mainpreprocess}
    If in some iteration of the while loop, after the swapping step, $v(\Bar{a}) < \frac{4}{5}$, the algorithm never runs out of items while performing bag-filling in line \ref{bagfilling}.
\end{restatable}

\begin{proof}
    Consider an arbitrary iteration of the while loop. Since the selected bags $\{a,b_1,\cdots, b_{k-1}\} \subseteq A$, by Claim \ref{cl:general}, we have $v(a) \geq 1 - \epsilon$ and $v(b_j) \geq 1 - \epsilon$ for all $j \in [k-1]$. If the value of $\Bar{a}$ drops below $\frac{4}{5}$ after swapping, it follows that:

    \begin{equation}\label{eq:aswap}
        v(\Bar{a})  = v(a) - v(H[1:k-1]) + v(G) < \frac{4}{5},
    \end{equation}

    where $H$ is the sorted list of HV items (in ascending order of value), and $G$ represents the set of highest-valued $k-1$ items added to $\Bar{a}$ from $\cup_{j=1}^{k-1} b_j$. Since $v(a) \geq 1 - \epsilon$ and $v(H[1:k-1]) \leq (k-1)v(H[k])$, we have:

    \[
    \frac{1}{k-1}\left(\frac{1}{5} - \epsilon + v(G)\right) < v(H[k]).
    \]

    Adding \( v(G) \) to both sides gives:

    \begin{equation}\label{eq:Hk}
        \frac{1}{k-1}\left(\frac{1}{5} - \epsilon\right) + \frac{k}{k-1}v(G) < v(H[k]) + v(G).
    \end{equation}

    Since \( H[k] \in a \setminus H[1:k-1] \), it follows that \( v(H[k]) \leq v(a) - v(H[1:k-1]) \). Using this inequality with equation \eqref{eq:Hk}, and then \eqref{eq:aswap}, we obtain:

    \begin{equation}\label{eq:vG}
        \frac{1}{k-1}\left(\frac{1}{5} - \epsilon\right) + \frac{k}{k-1}v(G) < \frac{4}{5}.
    \end{equation}

    Note that \( G \) contains the highest-valued \( k-1 \) items from \( \cup_{j=1}^{k-1} b_j \). For any \( g \in \cup_{j=1}^{k-1} b_j \setminus G \) and \( g' \in G \), we have \( v(g) \leq v(g') \), which implies \( (k-1)v(g) \leq v(G) \). 

    Additionally, the set \( P = \cup_{j=1}^{k-1} \Bar{b}_j \setminus H[1:k-1] \) corresponds exactly to \( \cup_{j=1}^{k-1} b_j \setminus G \), as the items outside \( G \) and \( H[1:k-1] \) remain unchanged during the swapping process. Therefore, for any \( g \in P \), it holds that \( (k-1)v(g) \leq v(G) \). Hence,

    \begin{equation}\label{eq:upperg}
        v(g) \leq \frac{1}{k}\left(\frac{4}{5} - \frac{(0.2 - \epsilon)}{k-1}\right).
    \end{equation}

    During bag filling from \( P \), the value of any bag \( \Bar{a} \) or \( b'_1, \dots, b'_{k-1} \) cannot exceed 
    $
    \frac{4}{5} + \frac{1}{k}\left(\frac{4}{5} - \frac{(0.2 - \epsilon)}{k-1}\right).
    $
    This is because, right before adding the last item, the bag was not yet claimed, meaning its value was less than \( \frac{4}{5} \). Furthermore, from Equation \eqref{eq:upperg}, the value of any item added from the pool \( P\) is bounded above, ensuring that the total value of the bag remains within the specified limit.

    Now we show that, after \( k-1 \) bags are claimed, sufficient value remains to construct the final bag such that its value exceeds \( \frac{4}{5} \).

    \textbf{Case 1:} \( k = 2 \).  
    From Equation \eqref{eq:upperg}, the value of any item in the remaining pool is at most \( 0.3 + \frac{\epsilon}{2} \). Therefore, the first claimed bag \( c \) can have a value of at most \( 1.1 + \frac{\epsilon}{2} \). Since \( v(a \cup b_1) \geq 2(1 - \epsilon) \), the remaining value is:
    \[
    v(a \cup b_1) - v(c) \geq 2(1 - \epsilon) - \left(1.1 + \frac{\epsilon}{2}\right).
    \]
    For \( \epsilon \leq 0.04 \), this exceeds \( \frac{4}{5} \), so the second bag is claimed.

    \textbf{Case 2:} \( k = 3 \). 
    From Equation \eqref{eq:upperg}, the value of any remaining item is at most \( \frac{0.7}{3} + \frac{\epsilon}{6} \). The first two claimed bags have a total value of at most $ 2\left(\frac{4}{5} + \frac{0.7}{3} + \frac{\epsilon}{6}\right).$ Since the total value of the initial bags $a, b_1, b_2$ is at least \( 3(1 - \epsilon) \), the remaining value for the last bag is at least:
    \[
    3(1 - \epsilon) - 2\left(\frac{4}{5} + \frac{0.7}{3} + \frac{\epsilon}{6}\right) = \frac{2.8 - 10\epsilon}{3}.
    \]
    For \( \epsilon \leq 0.04 \), this exceeds \( \frac{4}{5} \), so the last bag is claimed.
    
    \textbf{Case 3:} \( k \geq 4 \).
    From Equation \eqref{eq:upperg}, the value of any item is at most \( \frac{1}{5} \) since \( \epsilon < 0.2 \). The total value of the first \( k-1 \) claimed bags is at most $    (k-1)\cdot\left(\frac{4}{5} + \frac{1}{5}\right).$ Since the total value of the initial \( k \) bags ($a, b_1,\cdots, b_{k-1}$) exceeds \( k(1 - \epsilon) \), the remaining value for the last bag is at least:
    \[
    k(1 - \epsilon) - (k-1) = 1 - k\epsilon.
    \]
    Note that $k$ denotes the number of HV items in $a$, which cannot exceed $n$. 
    Also since \( \epsilon \leq \frac{1}{5n} \), the remaining value exceeds \( \frac{4}{5} \), so the last bag is claimed.

    \noindent Thus, the algorithm never runs out of valid items during bag-filling.
\end{proof}

\noindent Combining these results, we arrive at the following theorem.
\begin{theorem}
    Given a $1$-type $\text{ONI}_\alpha$ instance, Algorithm \ref{alg:SHV} returns an SHV \(\frac{4}{5}\)-MMS partition.
\end{theorem}
\begin{proof}
   While \( |A'| < n \), the while loop continues. In each iteration, if \( v(\bar{a}) < \frac{4}{5} \), then by \cref{lem:mainpreprocess} the bag‐filling procedure proceeds until \(\bar{a}\) and the \( k-1 \) bundles from \( B' \) are claimed; consequently, all \( k \) resulting bundles are added to \( A' \). Otherwise, at least one bundle—specifically, \(\bar{a}\)—is added to \( A' \). Hence, in each iteration at least one new bundle is added to \( A' \). Consequently, after at most linear number of iterations, \( |A'| = n \) and the algorithm terminates. By \cref{cl:general}, every bundle in \( A' \) contains exactly one HV item and has a value of at least \(\frac{4}{5}\). Hence, \( A' \) is an SHV \(\frac{4}{5}\)-MMS partition.
\end{proof}
\section{Improved MMS Approximations}\label{sec:4}

This section presents the main results, including Theorems \ref{thm:two} and \ref{thm:three}. Together with Lemma \ref{thm:oni-reduce} and Claim \ref{cl:ktype}, these results lead to the following general theorem:

\begin{theorem}
    For a $k$-type fair division instance $\I$, the $\text{ONI}_\alpha$ instance of $\I$ results in a $k'$-type instance with \(k' \leq k\), then:
    \begin{itemize}
        \item If \(k' = 2\), a \(\frac{4}{5}\)-MMS allocation exists.
        \item If \(k' = 3\), a \(\frac{16}{21}\)-MMS allocation exists.
    \end{itemize}
\end{theorem}

\begin{corollary}
    For any $2$-type fair division instance, a \(\frac{4}{5}\)-MMS allocation exists.
\end{corollary}

\begin{corollary}
    For any $3$-type fair division instance, a \(\frac{16}{21}\)-MMS allocation exists.
\end{corollary}

\subsection{$\frac{4}{5}$-MMS for $2$-type $\text{ONI}_\alpha$ Instance}\label{sec:two}

Given a $2$-type $\text{ONI}_\alpha$ instance, Algorithm \ref{alg:main2types} first constructs a $1$-type $\text{ONI}_\alpha$ instance \(\widehat{\mathcal{I}}\) based on the valuation function of the majority type. It computes the SHV \(\frac{4}{5}\)-MMS partition of \(\widehat{\mathcal{I}}\) and sorts the \(n\) bags in ascending order according to the minority type's valuation. The first \(T_1\) bags are assigned to type 1 agents. The remaining \(T_2\) bags are assigned to type 2 agents if all these bags are claimed by type 2 agents. If any bag is left unclaimed, the items from all \(T_2\) bags are pooled together. From this pool, each of the \(T_2\) HV items is placed into a new empty bag, and the bags are filled until each of them is claimed by type 2 agents. This process is called bag-filling. 
 
\begin{algorithm}[h]
\caption{$\frac{4}{5}$-MMS for $2$ types}\label{alg:main2types}
\begin{algorithmic}[1]
    \State \textbf{Input:} $2$-type $\text{ONI}_\alpha$ instance $\I = ([n], [m], \{T_1,T_2\}, \{v_1,v_2\})$, $\alpha = \frac{4}{5}$.
    \State \textbf{Output:} $\frac{4}{5}-$MMS. 
    \State Define a $1$-type $\text{ONI}_\alpha$ instance $\Icalhat = ([n],[m], \{n\}, \{v_1\})$
    \State Let $A$ be the SHV $\frac{4}{5}-$MMS partition after running Algorithm \ref{alg:SHV} on $\Icalhat$. 
     \State Sort bags of $A$ by type $2$ in an ascending order of valuation.
     \State Assign the first $T_1$ bags to type $1$ agents. 
     \If{type $2$ agents claim all remaining bags} 
     \State Assign all the remained bags to type $2$ agents. 
     \Else 
    \State $P\gets$ items from all remained bags.
     \State Put each HV item of $P$ in a new bag.
    \State From the remaining items in \( P \), fill each new bag until a type 2 agent claims it, then assign it to her.
     \EndIf
\end{algorithmic}
\end{algorithm}
\begin{theorem}\label{thm:two}
    Given a $2$-type $\text{ONI}_\alpha$ instance, Algorithm \ref{alg:main2types} returns a $\frac{4}{5}$-MMS.
\end{theorem}
\begin{proof}
Note that $\alpha= \frac{4}{5}$. In Algorithm \ref{alg:main2types}, first we obtain $A$, an SHV $\frac{4}{5}$-MMS partition of $\Icalhat$. Note that type $1$ agents claim all bags in $A$. Then type $2$ sorts all the $n$ bags of $A$ in an ascending order of valuation, therefore $v_{2}(A_{j}) \leq v_{2}(A_{j+1})$ for $j\in[n-1]$. Let $F = \{A_j\}^{T_1}_{j=1}$ be the collection of the first $T_1$ bags. Each bag in $F$ is assigned to a unique type $1$ agent. If $v_{2}(A_{T+1}) \geq \alpha$, since bags are sorted in ascending order, $v_{2}(A_j) \geq \alpha$ for all $j \in [T_1 +1,n]$. Therefore, all the remained bags are claimed by type $2$, and we'll assign them to the remained agents. In this case the assignment is a $\frac{4}{5}-$MMS as all the assigned bags are claimed by the corresponding agents. 

On the other hand if $v_{2}(A_{T+1}) < \alpha$, items in the remained bags are pooled in $P$, and since valuations are in an ascending order, $v_{2}(A_j) < \alpha$ for all $j \in [T_1]$. Therefore
\begin{equation}\label{eq:firstbound}
    \sum_{a\in F} v_{2}(a) < T_1\alpha.
\end{equation}
Note that since each bag in $A$ has exactly one HV item, we have exactly $n-T_1$ HV items in $P$, and as $T_2 = n-T_1$, the number of remained HV items is exactly the same as the number of type $2$ agents. We construct $T_2$ new empty bags, and by putting each HV item in a separate bag, the remained pool of items will not have any HV item any more. By Corollary \ref{col:bound} the value of each remained item in the pool is at most $\frac{\alpha}{2}$. When filling a bag until type 2 claims it, its value for type 2 cannot exceed \(\frac{3\alpha}{2}\) for the following reason; before the last item was added, the bag's value was less than \(\alpha\), and the last item contributes at most \(\frac{\alpha}{2}\). Let $A_{new}$ be the set of bags created during the bag filling phase. We have that
\begin{equation}\label{eq:secondbound}
    \sum_{a\in A_{new}}v_{2}(a) \leq T_{2}\frac{3\alpha}{2}.
\end{equation}
Set of bags assigned to all agents is $ \{A_j\}^{T_1}_{j=1} \cup A_{new}$. Putting  \cref{eq:firstbound} and \cref{eq:secondbound} together we obtain 
\begin{equation}
    \begin{aligned}
        \sum_{a\in \{A_j\}^{T_1}_{j=1} \cup A_{new}} v_{2}(a) & < T_1\alpha + T_{2} \frac{3\alpha}{2} \\
        & \leq \frac{n}{2}\alpha + \frac{n}{2}\cdot \frac{3\alpha}{2} = n
    \end{aligned}
\end{equation}

where we used $\alpha = \frac{4}{5}$, $T_1 + T_2 = n$ and $T_{1} \geq \frac{n}{2}$ since type $1$ agents are the majority by assumption. The total value of the assigned bags is upper bounded by $n$, the total available value in a normalized instance. Therefore, we never run out of goods while bag filling or, equivalently, all type $2$ agents receive a claimed bag. Therefore, this assignment is a $\frac{4}{5}-$MMS.  
\end{proof}

\subsection{$\frac{16}{21}-$MMS for $3$-type $\text{ONI}_\alpha$ Instance}\label{sec:three}

Given a $3$-type $\text{ONI}_\alpha$ instance, Algorithm \ref{alg:main3types} first constructs a $1$-type $\text{ONI}_\alpha$ instance \(\widehat{\mathcal{I}}\) based on the valuation function of the majority type, and computes the SHV \(\frac{4}{5}\)-MMS partition of \(\widehat{\mathcal{I}}\), $A$.  
For any choice of $\alpha$, each bag in $A$ is either liked by more than $\alpha$ by both types 2 and 3, or just one of them, or none. Based on the valuations of types 2 and 3, we partition the $n$ bags in $A$ into four classes $C_1,C_2,C_3$ and $C_4$ using Algorithm \ref{alg:organize}. By choosing $\alpha = \frac{16}{21}$, type $2$ agents claim all the bags in $C_2$ and $C_4$, type $3$ agents claim all the bags in $C_3$ and $C_4$, and type $1$ agents claim all the bags in $A=\bigcup^4_{i=1} C_i$. 

Finally, by comparing the number of bags in these classes with the number of agents of each type, we identify four cases. The specific approach for handling each case to achieve a $\frac{16}{21}$-MMS is detailed in Algorithm \ref{alg:main3types}. Notably, the last case addressed in this algorithm presents the greatest challenge and depends on executing Algorithm \ref{alg:case4}. We analyze each case of the algorithm separately and show that, in all four cases, Algorithm \ref{alg:main3types} guarantees a \( \frac{16}{21} \)-MMS allocation. Consequently, we derive the following theorem.
\begin{theorem}\label{thm:three}
    Given a $3$-type $\text{ONI}_\alpha$ instance, Algorithm \ref{alg:main3types} returns a $\frac{16}{21}$-MMS.
\end{theorem}

\begin{algorithm}[h]
\caption{$\frac{16}{21}$-MMS for $3$ types}\label{alg:main3types}
\begin{algorithmic}[1]
    \State \textbf{Input:} A $3$-type $\text{ONI}_\alpha$ instance $\I = ([n], [m], \{T_1,T_2, T_3\}, \{v_1,v_2,v_3\})$, $\alpha = \frac{16}{21}$.
    \State \textbf{Output:} $\frac{16}{21}-$MMS. 
    \State Define a $1$-type $\text{ONI}_\alpha$ instance $\Icalhat = ([n],[m], \{n\}, \{v_1\})$
    \State Let $A$ be the SHV $\frac{4}{5}-$MMS partition after running Algorithm \ref{alg:SHV} on $\Icalhat$.
    \State Using Algorithm \ref{alg:organize}, partition bags of $A$ into $C_1,C_2,C_3,C_4$.
    \If{$|C_2| > T_2$, $|C_3|> T_3$} \Comment{Case 1}
        \State Assign $T_2$ bags from $C_2$ to all type 2 agents.
        \State Assign $T_3$ bags from $C_3$ to all type $3$ agents. 
        \State Assign the remained bags to all type 1 agents. 
        \EndIf

     \If{$\exists i,i'\in\{2,3\} \text{s.t.} i\neq i', |C_i| > T_i$, $|C_{i'}|\leq T_{i'}$:} \Comment{Case 2}
     \State Assign $T_i$ bags of $C_i$ to all type $i$ agents. 
     \State Type $i'$ sorts all remained bags in an ascending order of valuation.
     \State Assign the first $T_1$ bags to type $1$ agents. 
     \If{all remained bags are claimed by type $i'$} 
     \State Assign all the remained bags to type $i'$ agents. 
     \Else 
    \State $P\gets$ items from all remained bags.
     \State Put each HV item of $P$ in a new bag.
     \State Fill each new bag with remaining items from \( P \) until a type \( i' \) agent claims it, then assign it to her.
     \EndIf
     \Else 
            \If{$|C_1|\leq T_1$}\Comment{Case 3}
                    \State Assign all bags in $C_1$ to some type 1 agents.
                    \State Assign all bags in $C_2$ to some type 2 agents.
                    \State Assign all bags in $C_3$ to some type $3$ agents.
                    \State Assign all remaining bags in \( C_4 \) to any remaining agents.
                    \EndIf
                    \If{$|C_1| > T_1$}\Comment{Case 4}
                        \State Assign $T_1$ bags in $C_1$ to all agents of type $1$. 
                        \State $P\gets$ items from all remained bags.
                        \State Let $H = HV\cap P$, $M=MV\cap P$, $L  = LV\cap P$. 
                        \State Run Algorithm \ref{alg:case4} with $\alpha = \frac{16}{21}$.
                    \EndIf
     \EndIf
\end{algorithmic}
\end{algorithm}

\begin{algorithm}
\caption{Clustering bags of $A$}\label{alg:organize}
\begin{algorithmic}[1]
    \State \textbf{Input:} $A= \{A_1, \cdots ,A_n\}$, $\alpha = \frac{16}{21}$.
    \State \textbf{Output:} A partition of bags in $A$ into 4 classes $C_1, C_2,C_3, C_4$.
    \For{$i\in\{1,2,3,4\}$}  $C_i \gets \emptyset$
    \EndFor
    \For{$a \in A$}
        \If{$v_2(a) < \alpha, v_3(a) < \alpha$} $C_1 \gets C_1 \cup \{a\}$.\EndIf
        \If{$v_2(a) \geq \alpha$, $v_3(a) < \alpha$} $C_2 \gets C_2\cup \{a\}$.\EndIf
        \If{$v_2(a) < \alpha$, $v_3(a) \geq \alpha$} $C_3 \gets C_3 \cup \{a\}$
        .\EndIf
        \If{$v_2(a) \geq \alpha$, $v_3(a) \geq \alpha$} $C_4 \gets C_4 \cup \{a\}$
        .\EndIf
    \EndFor
\end{algorithmic}
\end{algorithm}

\subsubsection{\textbf{Case 1: If $|C_2| > T_2$ and $|C_3|> T_3$}}
After allocating $T_2$ bags from $C_2$ to all type $2$ agents and $T_3$ bags from $C_3$ to all type $3$ agents, the remaining bags are assigned to all type $1$ agents. Since all bags in $A$ are claimed by type $1$ agents, all bags in $C_2$ by type $2$ agents, and all bags in $C_3$ by type $3$ agents, this allocation achieves a $\frac{16}{21}$-MMS guarantee.

\subsubsection{\textbf{Case 2: If $\exists i,i'\in\{2,3\} \text{ s.t. } i\neq i', |C_i| > T_i$ and $|C_{i'}|\leq T_{i'}$}}
Let \( A' \) denote the set of \( T_i \) arbitrary bags allocated to type \( i \) agents from \( C_i \). After this allocation, only two types of agents remain. Type 1 agents will receive \( T_1 \) bags from $A\setminus A'$. As each of the bags in $A$ has a value of at least \( \frac{4}{5} \) for type 1 agents, all of them are claimed by these agents. Since all the bags in \( A' \) are valued at less than \( \alpha \) by type \( i' \) agents, the loss incurred by type \( i' \) is minimal. Essentially, this scenario is equivalent to considering type \( i' \) agents forfeiting these bags to type \( 1 \) agents in a two type setting where the number of type $1$ agents is $T_1+T_i$. 
\begin{restatable}{lemma}{lemcase}
\label{lem:case2}
If $\exists i,i'\in\{2,3\} \text{ s.t. } i\neq i', |C_i| > T_i$ and $|C_{i'}|\leq T_{i'}$, Algorithm \ref{alg:main3types} returns a $\frac{16}{21}$-MMS. 
\end{restatable}

\begin{proof}
Let $A_1 ,\cdots , A_{n-T_i}$ be the remained bags of $A$ after $T_i$ bags of $C_i$ are assigned to all type $i$ agents, and type $i'$ sorted the rest of the $n-T_i$ bags in an ascending order of valuation i.e. $v_{i'}(A_{j}) \leq v_{i'}(A_{j+1})$ for $j\in[n-T_{i}-1]$. 

The first $T_1$ bags are assigned to type $1$ agents, and the remained bags will be 
$$A'' = \{A_j\}^{n-T_i}_{j=T_1+1}.$$
If $v_{i'}(A_{T+1}) \geq \alpha$, since bags are sorted in ascending order, $v_{i'}(a) \geq \alpha$ for all $a\in A''$. Therefore, all the remained bags are claimed by type $i'$, and we'll assign them to type $i'$ agents. In this case, all the assigned bags are claimed by the corresponding agents, and therefore the assignment is a $\frac{16}{21}$- MMS. 

On the other hand if $v_{i'}(A_{T+1}) < \alpha$, as valuations are in an ascending order, $v_{i'}(A_j) \leq v_{i'}(A_{T_1+1})$ for all $j \in [T_1]$. Hence, $v_{i'}(A_j) < \alpha$ for all $j \in [T_1]$, therefore all the bags assigned to type $1$ agents have value of less than $\alpha$ for type $i'$. Note that as bags that were assigned to type $i$ agents were from $C_i$, and bags in $C_i$ have a value of less than $\alpha$ for type $i'$(when $\{i,i'\} \subseteq \{2,3\}$), any bag in $A\setminus A''$ is valued less than $\alpha$ for type $i'$. As $|A\setminus A''| = T_1 + T_i = n-T_{i'}$, 
\begin{equation}\label{ineq:first}
    \sum_{a\in A\setminus A''} v_{i'}(a) < (n-T_{i'})\alpha.
\end{equation}
Let \( P \) denote the set of remaining items after pooling all the items from the bags in \( A'' \). 
Since each bag in \( A'' \) contains exactly one HV item, there are precisely \( T_{i'} \) HV items in \( P \), 
and we require exactly \( T_{i'} \) bags for type \( i' \) agents. By placing each HV item into a separate new bag, 
the remaining items in \( P \) will no longer contain any HV items. Hence, according to Corollary \ref{col:bound}, 
the value of each remaining item in the pool is at most \( \frac{\alpha}{2} \).

Let \( A_{new} \) denote the set of newly created bags. Performing bag filling for each bag in \( A_{new} \) 
until it is claimed by a type \( i' \) agent ensures that the total value of each bag does not exceed \( \frac{3\alpha}{2} \), 
as the value of the last added item can be at most \( \frac{\alpha}{2} \). Therefore,

\begin{equation}\label{ineq:second}
    \sum_{a\in A_{new}}v_{i'}(a) \leq T_{i'}\frac{3\alpha}{2}.
\end{equation}
Set of bags assigned to all agents is $(A\setminus A'') \cup A_{new}$. Putting inequalities \eqref{ineq:second} and \eqref{ineq:first} together we obtain
\begin{equation}
    \begin{aligned}
        \sum_{a\in (A\setminus A'') \cup A_{new}} v_{i'}(a) & < (n-T_{i'})\alpha + T_{i'} \frac{3\alpha}{2}.
    \end{aligned}
\end{equation}
Since type $1$ is the majority type, and $i'\in\{2,3\}$, $T_{i'} \leq T_1$. Since the total number of agents is $n$, $T_{i'}$ cannot be greater than $\frac{n}{2}$. Therefore,
\begin{equation}
    \begin{aligned}
        \sum_{a\in (A\setminus A'') \cup A_{new}} v_{i'}(a)
        & \leq \frac{n}{2}\alpha + \frac{n}{2}\cdot \frac{3\alpha}{2} < n
    \end{aligned}
\end{equation}
where we used $\alpha = \frac{16}{21}$. Hence, we never run out of goods in the bag filling phase or, equivalently, all type $i'$ agents receive a claimed bag at some point during the algorithm. Note that since type $i$ claims all bags in $C_i$, and all bags in $A$ are claimed by type $1$ agents, this assignment is a $\frac{16}{21}$-MMS. 
\end{proof}

\subsubsection{\textbf{Case 3: If $|C_2| \leq T_2$, and $|C_3|\leq T_3$, $|C_1|\leq T_1$}}
After assigning all bags in $C_1$, $C_2, C_3$ to some agents of types $1,2,3$ respectively, we are remained with bags in $C_4$ that are claimed by all types. Hence we can assign them to any remaining agents, and obtain a $\frac{16}{21}$-MMS assignment. 

\subsubsection{\textbf{Case 4: If $|C_2| \leq T_2$, and $|C_3|\leq T_3$, $|C_1|> T_1$}}
\begin{definition}
    A type is considered \textit{saturated} if every agent of that type has received a claimed bag in the assignment; otherwise, it is \textit{unsaturated}.
\end{definition}
\begin{definition}
     A bag $B$ is \emph{safe} for a type $i$ agents if $v_i(B)\leq 1.$
\end{definition}
After assigning $T_1$ bags in $C_1$ to all agents of type $1$, type $1$ is saturated. Let \( H \), \( M \), and \( L \) denote the lists of high-valued, middle-valued, and low-valued items, respectively, that remain in the remaining \( n - T_1 \) bags. Each list is ordered in descending order of valuation.
Since each assigned bag to type $1$ had exactly one HV item, $|H| = n-T_1$. On the other hand, since any number of middle-valued items might remain after the assignment of bags to the majority type, $0 \leq|M|\leq n$. Finally, \cref{alg:case4} is invoked.  \\

\textbf{Main Ideas of Algorithm \ref{alg:case4}:}
The algorithm consists of two parts. Initially, both type 2 and type 3 are active and unsaturated. A type remains active until it meets one of the conditions specified in lines \ref{line:inactive1} or \ref{line:inactive2}, or until it becomes saturated in the first part of the algorithm (by line \ref{line:inactive3}). Once a type becomes inactive, it stays inactive; if none of these conditions are met, it may never become inactive.

As the algorithm proceeds, items are grouped into bundles that are either immediately assigned to agents or saved for future assignment. In either case, the bundled items become unavailable and are removed from $H$, $M$, and $L$, thereby reducing the number of available items. Consequently, at any given stage, $H$, $M$, and $L$ represent the currently available high-valued, middle-valued, and low-valued items. The algorithm operates in two parts, as described below.

\begin{enumerate}
    \item \textbf{First Part:}  
    As long as there is an active type and remaining high-valued and middle-valued items, the algorithm constructs certain bundles that are \emph{safe} for all active types. Each bundle contains one HV item and one MV item. The bags are either:
    \begin{itemize}
        \item Assigned to an agent from an unsaturated type that claims the bag, or
        \item Saved in \( F \) for the second part if no unsaturated type claims the bag.
    \end{itemize}
    This process continues until no high-valued or middle-valued items remain, or no active types remain. Given a set of available high-valued and middle-valued items \( H \) and \( M \), where \( H \neq \emptyset \) and \( M \neq \emptyset \), type \( i \) becomes inactive (i.e., is no longer active) if it becomes saturated (by line \ref{line:inactive3}), or if \( v_i(M[1],H[-1]) > 1 \) and there exists no \( j \) such that \( \alpha \leq v_i(M[j], H[-1]) \leq 1 \), by lines \ref{line:inactive1} or \ref{line:inactive2}. Recall that $H[-1]$ is the least valued available HV item.  

    \item \textbf{Second Part:}  
    When both types are inactive or either \( M = \emptyset \) or \( H = \emptyset \), the second part begins. The collection of bags \( F \) inherited from the first part consists of bags, each containing one HV and one MV item, with a total value of less than \( \alpha \) for any unsaturated type. If \( H \neq \emptyset \), each remaining item in \( H \) is placed in a separate new bag and the bag is added to \( F \).

    The algorithm then picks an arbitrary bag from \( F \) and fills it with remaining items until an unsaturated type claims it. The bag is then assigned to an arbitrary agent belonging to an unsaturated type that claimed it. This bag-filling procedure repeats until all types are saturated. 
\end{enumerate}

Let us illustrates the first part of \Cref{alg:case4} on the following example.

\begin{example}\label{example1}
Consider an instance with $9$ agents, where every group of three agents belongs to the same type. Under our assumption, the instance is $ONI_{\alpha}$ and is therefore ordered. In particular, the set of high-valued (HV) items is 
\(
\text{HV} = \{1, 2, \ldots, 9\},
\)
the set of middle-valued (MV) items is 
\(
\text{MV} = \{10, 11, \ldots, 18\},
\)
and all remaining items are low valued.

Each type 1 agent has taken a bundle containing one high-valued item. Suppose that items $\{2,4,9\}$ are the high-valued items assigned to type 1 agents, and assume that none of these agents have taken any middle-valued item. In \cref{fig:example}, only the high-valued and middle-valued items are displayed, with the items taken by type 1 agents crossed out.

Next, consider the valuations for types 2 and 3. Type 3 agents assign an equal value of $\frac{\alpha}{3}$ to each of the top $18$ items, while the valuation function for type 2 agents over these items is as follows:
\[
v_2(j)=
\begin{cases}
\alpha-\epsilon, & \text{if } j\in\{1,\ldots,5\}, \\
\frac{\alpha}{2}+2\epsilon, & \text{if } j=6, \\
\frac{\alpha}{2}-\epsilon, & \text{if } j\in\{7,\ldots,13\}, \\
1-\alpha-\epsilon, & \text{if } j=14, \\
\frac{\epsilon}{2}, & \text{if } j\in\{15,\ldots,18\}.
\end{cases}
\]

Note that any bundle containing only two items is valued below $\alpha$ for type~3 agents. Consequently, no bag constructed in the first part of \cref{alg:case4} is ever claimed by a type~3 agent.

Initially, both types~2 and~3 are unsaturated and active. In each iteration of the while loop in the first part of \cref{alg:case4}, the algorithm constructs a bundle by pairing the least valued available high-valued item with the highest valued available middle-valued item. 

At the outset, the algorithm constructs the bundle
\(
B_1 = \{8,10\}.
\)
Since $B_1$ is valued below $\alpha$ for both active types, it is saved in the set $F$. In the next iteration, the bundle
\(
B_2 = \{7,11\}
\)
is constructed, and as it too is valued below $\alpha$ for both active types, it is added to $F$.

Subsequently, the algorithm constructs the bundle
\(
B_3 = \{6,12\}.
\)
This bundle is valued above $\alpha$ but below $1$ for type~2, while it remains below $\alpha$ for type~3; hence, $B_3$ is allocated to a type~2 agent.

Next, the algorithm forms the bundle
\(
B_4 = \{5,13\}.
\)
Since $B_4$ is valued above $1$ for type~2, the algorithm searches for an available MV item that, when paired with item $5$, yields a bundle whose value for type~2 lies between $\alpha$ and $1$. The bundle
\(
B_5 = \{5,14\}
\)
satisfies this condition and is thus allocated to a type~2 agent.

The algorithm then constructs the bundle
\(
B_6 = \{3,13\}.
\)
Again, as $B_6$ is valued above $1$ for type~2, the algorithm seeks an available MV item that, when paired with item $3$, results in a value between $\alpha$ and $1$ for type~2. Since no such middle-valued item exists, type~2 becomes inactive at this point.

In the next iteration of the while loop, $B_6$ is constructed once more and assigned to the last remaining type~2 agent. With all three type~2 agents now having received a bundle, type~2 is saturated.

Subsequently, the algorithm constructs the bundle
\(
B_7 = \{1,15\}.
\)
Since $B_7$ is valued below $\alpha$ by the only active and unsaturated type (type~3), it is added to $F$. At this stage, no further high-valued items remain, and the algorithm proceeds to part~2. It is important to note that type~3 never becomes inactive.

At the beginning of part~2, we have
\(
F = \{B_1, B_2, B_7\},
\)
and no available high-valued items remain. Each bag in $F$ is subsequently augmented with the remaining items until a type~3 agent claims it, at which point the bag is allocated to that type~3 agent.\qed

\begin{figure}
    \centering
    \begin{tikzpicture}[every node/.style={draw, circle, minimum size=1cm, inner sep=0pt}]
  \foreach \i in {1,...,9} {
    \node (b\i) at ({(\i-1)*1.5},0) {\i};
  }
  
  \foreach \j in {1,...,9} {
    \pgfmathtruncatemacro{\num}{19-\j}
    \node (t\j) at ({(\j-1)*1.5},2) {\num};
  }
  
  \foreach \i in {2,4,9} {
    \draw[line width=1pt] (b\i.north east) -- (b\i.south west);
    \draw[line width=1pt] (b\i.north west) -- (b\i.south east);
  }
  
    \draw[red, thick] (11.3,1) ellipse [x radius=2.2, y radius=0.65, rotate=53, xshift =0pt];

    \draw[red, thick] (9.8,1) ellipse [x radius=2.2, y radius=0.65, rotate=53, xshift =0pt];

    \draw[green, thick] (8.3,1) ellipse [x radius=2.2, y radius=0.65, rotate=53, xshift =0pt];

    \draw[green, thick] (6,1) ellipse [x radius=2.2, y radius=0.65, rotate = 90, xshift =0pt];

    \draw[green, thick] (5.2,1) ellipse [x radius=3.5, y radius=0.65, rotate = 25, xshift =0pt];

    \draw[red, thick] (2.2,1) ellipse [x radius=3.7, y radius=0.65, rotate = 25, xshift =0pt];

    \node[draw=none] at (-1.7, 0) {\textbf{HV :}};

    \node[draw=none] at (-1.7, 2) {\textbf{MV :}};

\end{tikzpicture}
    \caption{This figure illustrates Example \ref{example1}. The red bundles are saved in $F$, and the green bundles are allocated to type $2$ agents. }
    \label{fig:example}
\end{figure}
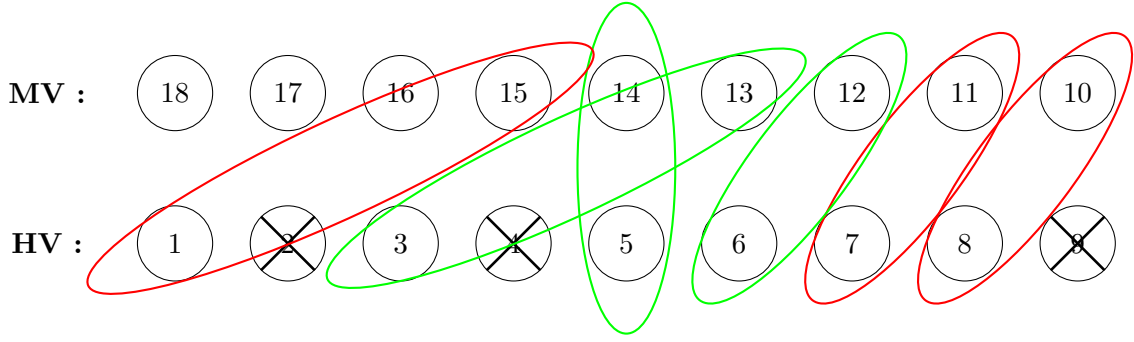

\end{example}

\begin{algorithm}[h]
\caption{Algorithm for Case 4}\label{alg:case4}
\begin{algorithmic}[1]
    \State Let $F\gets \emptyset$, $\mathtt{active} \gets \{2,3\}$, $\mathtt{unsaturated}\gets \{2,3\}$.\\\Comment{Part 1 begins\hspace{5in}}
    \While{$H\neq \emptyset$ and $M\neq \emptyset$ and $\mathtt{active}\neq \emptyset$}
        \State $B \gets \{M[1], H[-1]\}$.
        \State Let \( X = \{ i \in \mathtt{active} \mid v_i(B) > 1 \} \).
        \If{$|X| = 0$} 
            \If{$\exists j\in \mathtt{unsaturated}, v_j(B) \geq \alpha$} 
            \State Assign B to a type $j$ agent. 
            \Else, $F\gets F\cup B$. 
        \EndIf
        \EndIf
        \If{$|X| = 1$}
            \State Let $i$ be the unique element of $X$.
            \If{$\exists j$ s.t. $\alpha \leq v_i(M[j], H[-1]) \leq 1$}
                \State Assign $B' = \{M[j], H[-1]\}$ to a type $i$ agent.
             \Else, $\mathtt{active} \gets \mathtt{active}\setminus \{i\}$. \label{line:inactive1}
             \EndIf     
        \EndIf
        \If{$|X| = 2$}
            \If{$\exists i\in X$, s.t. $\nexists j: \alpha \leq v_i(M[j],H[-1]) \leq 1$}
                \State $\mathtt{active} \gets \mathtt{active}\setminus \{i\}$.\label{line:inactive2}
            \Else 
                \State Let $j$ be the largest index where for some $i\in X$, $\alpha \leq v_i(M[j], H[-1]) \leq 1$.
                \State Assign $B' = \{M[j], H[-1]\}$ to type $i$.
            \EndIf    
        \EndIf
        \State If a bag is saved or assigned, remove its items from $M$ and $H$.
        \State If a type is saturated, remove it from $\mathtt{active}$ and $\mathtt{unsaturated}$.\label{line:inactive3} \\ \Comment{Part 2 begins \hspace{5in}}
    \EndWhile
    \State Place all the remaining items of $H$ in a separate new bag and add the bag to $F$.
    \State Let $R=M \cup L$ be the pool of remained items.
    \While{$\mathtt{unsaturated} \neq \emptyset$}
        \State Choose a bag \( B \) from \( F \), fill it with available items from \( R \), and assign it to any agent of type in \( \mathtt{unsaturated} \) upon claim.
        \State Update $F$, $R$, and $\mathtt{unsaturated}$.
    \EndWhile  
\end{algorithmic}
\end{algorithm}

The following proposition is a key component of the analysis. 
\begin{restatable}{proposition}{lemsafe}
\label{lem:safe}
    Throughout the first part of the algorithm, as long as a type remains active, all saved and assigned bags are guaranteed to be safe for that type.
\end{restatable}
\begin{proof}
    Let us consider an arbitrary iteration of the first while loop where type $i$ is active. Let $H$ and $M$ be the set of available high-valued and middle-valued items respectively. Note that $B = \{M[1],H[-1]\}$. Let us consider all different possibilities and prove that in any case, if a bag is assigned or save, it is safe for type $i$.  
    
    \paragraph{}If type $i$ is the only active type, and $v_i(B) < 1$, either this bag is assigned or saved in $F$. In either cases, bag $B$ is safe for it. If $v_i(B) > 1$, either type $i$ becomes inactive or a bag $B'$ where $\alpha \leq v_i(B') \leq 1$ is assigned. In this case $B'$ is safe for type $i$.

    \paragraph{} If type $i$ is \textbf{not} the only active type, the other active type is $i' = \{2,3\}\setminus i$. If for both active types value of $B$ is less than $1$, $B$ is either saved or assigned, and in both cases $B$ is safe for type $i$. 
    
    If type $i$ is the only active type that values $B$ more than $1$, either it gets inactive or a bag $B'$ where $\alpha \leq v_i(B') \leq 1$ is assigned. Therefore $B'$ is safe for type $i$.
    
    If type $i'$ is the only active type that values $B$ more than $1$, either type $i'$ gets inactive or a bag \( B' = \{M[j], H[-1]\} \) is assigned for some \( j > 1 \), satisfying the condition \( \alpha \leq v_{i'}(B') \leq 1 \). Note that by assumption $v_i(B) < 1$ where $B = \{M[1],H[-1]\}$. Since $M$ is an ordered list, $v_i(M[j]) \leq v_i(M[1])$, therefore $v_i(B') \leq v_i(B)$. Consequently we obtain $v_i(B')<1$ implying $B'$ is a safe bag for type $i$.
    
    If both types $i$, $i'$ value $B$ more than $1$, either one of the types get inactive, or $\exists j,j'$ where $\alpha \leq v_i(M[j],H[-1])\leq 1$, and $\alpha \leq v_{i'}(M[j'], H[-1]) \leq 1$. Consider the largest index $j''$ where one of the active types values $B' = \{M[j''],H[-1]\}$ between $\alpha$ and $1$. The algorithm will assign $B'$ in this iteration. Suppose for contradiction that $B'$ is not safe for type $i$, hence $v_i(M[j''],H[-1]) > 1$. Since $\exists j$ where $\alpha \leq v_i(M[j],H[-1])\leq 1$, it follows that $v_i(M[j]) \leq v_i(M[j''])$. As \( M \) is ordered in descending order of valuation, \( M[j] \) must appear later than \( M[j''] \) in the list. Therefore $j > j''$, implying $\exists j > j''$ such that $\alpha \leq v_i(M[j],H[-1])\leq 1$, and this is a contradiction with the assumption that $j''$ is the largest index. Therefore $B'$ should be safe for type $i$. 
    
\end{proof}

We now present the key lemma required for the proof of \cref{lem:fourthcase}. Due to the nuanced complexity of its proof, we defer the detailed argument to \cref{app:case4} to maintain a smooth exposition.

\begin{lemma}\label{thm:case4}
    In Algorithm \ref{alg:case4}, both types $2$ and $3$ become saturated.
\end{lemma}
Finally, in the next lemma, we show that even in the fourth case \cref{alg:main3types} guarantees a $\frac{16}{21}$-MMS allocation. 
\begin{lemma}\label{lem:fourthcase}
    If $|C_2| \leq T_2$, and $|C_3|\leq T_3$, and $|C_1|> T_1$, \cref{alg:main3types} guarantees a \( \frac{16}{21} \)-MMS allocation.
\end{lemma}
\begin{proof}
    Notably, each bag saved in \( F \) or assigned to an agent in the first part of \cref{alg:case4} contains exactly one HV item, and the total number of HV items is \( n \). Consequently, before the bag-filling process begins, the number of bags in \( F \) is equal to the number of remaining HV items, which is \( n \) minus the number of HV items already assigned. Since each of the assigned bags has exactly one HV item, the number of bags in \( F \) precisely equals \( n \) minus the number of assigned bags, which corresponds to the number of agents who have not yet received a bag. Thus, to guarantee that every agent receives a claimed bag, it suffices to show that the supply of goods remains sufficient throughout the bag-filling process for the bags in \( F \). This is established in lemma \ref{thm:case4}. Combining this with the fact that all bags assigned to type 1 agents have a value of at least \( \frac{4}{5} \), it follows that Algorithm \ref{alg:main3types} guarantees a \( \frac{16}{21} \)-MMS allocation in the fourth case. 
\end{proof}

\section{Conclusion}
In the fair division of indivisible goods, the maximin share is one of the most extensively studied fairness notions. Determining tight lower and upper bounds on the maximum $\alpha$ for which $\alpha$-MMS allocations are guaranteed to exist remains a fundamental open problem. Since MMS allocations do not always exist even for instances with two agent types, our improved bounds represent a significant step forward in this area. To gain deeper insights into this problem, we focused on the special cases of two and three agent types. A compelling open question is whether uniform improvements can be achieved for any $k<n$ types, with guarantees that decreases as $k$ increases, surpassing the current best-known approximation of \(\frac{3}{4} + \frac{3}{3836}\).

\bibliographystyle{alpha}
\bibliography{main}
\clearpage    
\appendix
\section{Proof of Lemma \ref{thm:case4}}\label{app:case4}
If the second while loop begins with \( F = \emptyset \), it implies that all HV items have already been assigned. Since each HV item is assigned to a unique agent, and the total number of HV items is $n$, it follows that all agents have received a bag. Furthermore, as part 1 of the algorithm assigns a bag only to one of the agents who claimed it, this ensures that all \( n \) agents have received a bag they claimed, meaning all types are saturated.

Now, consider the case where the second while loop begins with \( F \neq \emptyset \). As previously established, the number of bags in \( F \) equals the number of agents who have not yet been assigned a bag. It remains to demonstrate that during the bag-filling process, the supply of goods is sufficient to fill all bags in \( F \) such that each is claimed by agents from an unsaturated type.

Assume, for contradiction, that for some type \( i \in \{2,3\} \), type \( i \) is not saturated by the end of the algorithm. This could only occur if the bag-filling process exhausts the available goods, preventing the allocation of bags with a value of at least \( \frac{16}{21} \) to all agents of type \( i \). However, we will show that such a situation never arises. In all possible cases, the total value of the \( n \) assigned bags cannot exceed the total available value of \( n \), thereby ensuring that type \( i \) is saturated.

Let \( P_i \) denote the set of all bags that were assigned or saved during the \textbf{first} part of the algorithm while type \( i \) was \textbf{active}, and let \( F_i \subseteq P_i \) represent the subset of bags that were specifically saved in \( F \) during this period. Proposition~\ref{lem:safe} shows that $\forall B\in P_i, v_i(B) \leq 1$. We can partition bags in $P_i$ as the following: 
\begin{enumerate}
    \item $P^1_i$ = \{$B\in F_i \mid v_i(B) < \alpha$\}. 
    \item $P^2_i$ = \{$B\in P_i\setminus F_i \mid v_i(B) < \alpha$ \}
    \item $P^3_i$ = \{$B\in P_i \mid \alpha \leq v_i(B) \leq 1$ \}
\end{enumerate}
$P^1_i$ contains the bags that were \emph{saved} in the first part of the algorithm while type $i$ was active. Note that according to the algorithm a bag is saved if and only if its value is less than $\alpha$ for all unsaturated types. Let $b_i = |P^1_i|$. 

$P^2_i$ contains the bags that were \emph{assigned} in the first part of the algorithm while type $i$ was active, where each bag had a value of less than $\alpha$ for type $i$. Note that these bags were assigned because their value was more than $\alpha$ for another unsaturated type $i'$, where $i'\neq i$. Let $\tau_i = |P^2_i|$. 

$P^3_i$ contains the bags that were \emph{assigned} in the first part while type $i$ was active, where each bag had a value of more than $\alpha$ for type $i$. Note that since the value of each bag in \( P^3_i \) exceeds \( \alpha \) for an unsaturated type, it is not included in \( F_i \). Let $k_i = |P^3_i|$.\\

\begin{example}
    Recall that in \cref{example1}, at the moment when type~2 became inactive, the state of the HV and MV items was as shown in \cref{fig:example2}. In that figure, crossed items denote those taken by type~1 agents, red bundles represent saved bundles, and green bundles indicate bundles that have been allocated to agents. Since the number of saved bundles in $F$ is $2$, we have $b_2 = 2$. Moreover, because the green bundles allocated to type~2 agents have values in the range $[\alpha,1]$ for type~2, we set $k_2 = 2$. Finally, as none of the assigned bundles were valued below $\alpha$ for type~2, it follows that $\tau_2 = 0$.

    Since type~3 never became inactive, it remained active throughout the execution of part~1 of \cref{alg:case4} (see \cref{fig:example} for an illustration of the bundles constructed during this phase). For type~3, the three red bundles saved in $F$ during the first while loop yield $b_3 = 3$. Furthermore, because all three green bundles allocated to type~2 agents were valued below $\alpha$ for type~3, we have $\tau_3 = 3$. Finally, as none of the assigned bundles were valued in the range $[\alpha,1]$ for type~3, it follows that $k_3 = 0$. \qed
\end{example}

Let $\{B_j\}^{T_1}_{j=1}$ be the first $T_1$ bags assigned to type $1$ agents. Since we are discussing Case 4 of Algorithm \ref{alg:main3types}, all bags in $\{B_j\}^{T_1}_{j=1}$ had a value of less than $\alpha$ for type $i$, therefore 
\begin{equation}\label{eq:firstT}
    v_i(\cup^{T_1}_{j=1} B_j) \leq T_1\alpha.
\end{equation}

Note that $P_i\setminus F_i$ is the set of all bags that were \emph{assigned} in the first part of the algorithm while type $i$ was active. Since $P_i\setminus F_i = P^2_i \cup P^3_i$, we have that $|P_i\setminus F_i| = \tau_i + k_i$. Let $\{B_j\}^{T_1+\tau_i + k_i}_{j=T_1+1}$ represent the set of bags in $P_i\setminus F_i$. From this set, $\tau_i$ bags had a value of less than $\alpha$ for type $i$, and $k_i$ of them had a value between $\alpha$ and $1$ for type $i$. Hence,
\begin{equation}\label{eq:secondtau}
    v_i(\cup^{T_1+\tau_i + k_i}_{j=T_1+1} B_j) \leq \tau_i \alpha + k_i.
\end{equation}
Putting \eqref{eq:firstT} and \eqref{eq:secondtau} together, 
\begin{equation}\label{eq:boundonfirst}
    v_i(\cup^{T_1+\tau_i+k_i}_{j=1} B_j) \leq (T_1+\tau_i)\alpha + k_i.
\end{equation}
Since we are focusing on just one type, for simplicity, let $T = T_1 + \tau_i$, $k = k_i$ and $b = b_i$. Since $T_1 \geq \frac{n}{3}$, we have $T \geq \frac{n}{3}$. Hence, we can rewrite \eqref{eq:boundonfirst} as follows:
\begin{equation}\label{eq:simple}
    v_i(\cup^{T+k}_{j=1} B_j) \leq T\alpha + k.
\end{equation}

\paragraph{$\bullet$} If type $i$ never gets inactive, and the first while loop terminates because $M = \emptyset$, when part 2 begins, $M = \emptyset$, implying $R = L$. For any assigned bag $B_j$ in the second part of the algorithm, right before adding the last item, the value of the bag was less than $\alpha$ for all unsaturated types, including $i$. Since the last added item is from $R$ and each item in $R$ has a value of at most $\frac{\alpha}{3}$ by Corollary \ref{col:bound}, the value of each assigned bag in part 2 is at most $\frac{4\alpha}{3}$. Since the number of bags assigned in the first part is \( \tau_i + k_i \), the number of bags that remain to be assigned in part 2 is given by \( n - T_1 - \tau_i - k \), which simplifies to \( n - T - k \). Let $\{B_j\}^n_{j=T+k+1}$ be the collection of bags assigned in part 2; $v_i(\cup^n_{j=T+k+1} B_j) \leq (n-T-k)\frac{4\alpha}{3}$. By combining this with \eqref{eq:simple}, we derive
\begin{equation}
    \begin{aligned}
        v_i(\cup^n_{j=1}B_j) & \leq T\alpha + k + (n-T-k)\frac{4\alpha}{3}\\
        & = n\frac{4\alpha}{3} -T\frac{\alpha}{3} + k(1-\frac{4\alpha}{3}) \\
        & \leq n\frac{4\alpha}{3} -\frac{n}{3}\frac{\alpha}{3} = n\frac{176}{189} < n, 
    \end{aligned}
\end{equation}
where we used the fact that $T\geq \frac{n}{3}$ and $\alpha =\frac{16}{21}$.

\paragraph{$\bullet$}If type \( i \) never becomes inactive and the first while loop terminates due to \( H = \emptyset \), then throughout this loop, each of the initial \( n - T_1 \) HV items in \( H \) is paired with an MV item. Each pair forms a bundle that is either assigned or saved in $F$. Since exactly \( \tau_i + k_i \) pairs were assigned during the first while loop, the remaining \( n - T_1 - \tau_i - k_i \) pairs are saved in \( F \). By assumption, when part 2 begins, \( F \neq \emptyset \); therefore, in part 2, each bag in \( F \) is filled until it is claimed by an unsaturated type.

Observe that during the first while loop, $n - T_1$ MV items were either saved or assigned in bundles containing one HV item and one MV item. Given that there are at most $n$ middle-valued items available at the beginning, at most $T_1$ of them remain in $R$. Consequently, when the second while loop begins, we have $|M| \leq T_1$. According to Corollary \ref{col:bound}, each of these items has a value of at most $\alpha/2$. During the bag-filling process for the bags in $F$, each item in $M$ can be added as the last item in a bag. Therefore, at most $\min(T_1, n - T - k)$ bags may contain items with a total value of up to $\frac{3\alpha}{2}$. Let $\{B_j\}_{j=T+k+1}^{n}$ represent the collection of bags assigned in part 2. Now, consider the following cases:
\begin{itemize}
    \item If $n-T-k \leq T_1$, during the bag filling at most all of the $n-T-k$ bags have a value of at most $\frac{3\alpha}{2}$, so $v_i(\cup^n_{j=T+k+1} B_j) \leq (n-T-k)\frac{3\alpha}{2}$. By combining this with \eqref{eq:simple}, we derive 
    \begin{equation}
    \begin{aligned}
        v_i(\cup^n_{j=1}B_j) & \leq T\alpha + k + (n-T-k)\frac{3\alpha}{2}\\
        & = n\frac{3\alpha}{2} -T\frac{\alpha}{2} + k(1-\frac{3\alpha}{2}) \\
        & = n\frac{3\alpha}{2} -T_1\frac{\alpha}{2} - \tau_i\frac{\alpha}{2} + k(1-\frac{3\alpha}{2}) 
    \end{aligned}
    \end{equation}
By assumption, $n-T-k \leq T_1$, so $n\leq 2T_1 + \tau_i + k$. As $T_1\geq \frac{n}{3}$, the upper bound is obtained when $T_1 = \frac{n}{3}$, $\tau_i = 0$, $k=\frac{n}{3}$. Therefore, 
    \begin{equation}
    \begin{aligned}
        v_i(\cup^n_{j=1}B_j) & \leq
        n\frac{3\alpha}{2} -\frac{n}{3}\frac{\alpha}{2} + \frac{n}{3}(1-\frac{3\alpha}{2}) \\& = n\frac{61}{63} \leq n
    \end{aligned}
    \end{equation}
where we used $\alpha =\frac{16}{21}$.

\item If $n - T - k > T_1$, then during the bag-filling process, at most $T_1$ bags can reach a value of up to $\frac{3\alpha}{2}$. The remaining $n - T - k - T_1$ bags, however, can have a value of at most $\frac{4\alpha}{3}$, since, by Corollary \ref{col:bound}, the value of any item in $L$ is upper bounded by $\frac{\alpha}{3}$; hence $v_i(\cup^n_{j=T+k+1} B_j) \leq (n-T-k-T_1)\frac{4\alpha}{3} + T_1\frac{3\alpha}{2}$. By combining this with \eqref{eq:simple}, we derive 
    \begin{equation}
    \begin{aligned}
        v_i(\cup^n_{j=1}B_j) & \leq T\alpha + k + (n-T-k-T_1)\frac{4\alpha}{3} + T_1\frac{3\alpha}{2}\\
        & = n\frac{4\alpha}{3} - T\frac{\alpha}{3} + T_1\frac{\alpha}{6} + k(1-\frac{4\alpha}{3}) \\
        & \leq n\frac{4\alpha}{3} - T\frac{\alpha}{3} + T_1\frac{\alpha}{6}\\
        & = n\frac{4\alpha}{3} - \tau_i\frac{\alpha}{3} - T_1\frac{\alpha}{6}\\
        & \leq n\frac{4\alpha}{3} - T_1\frac{\alpha}{6} \\
        & \leq n\frac{4\alpha}{3} - \frac{n}{3}\frac{\alpha}{6} \\
        & = n\frac{184}{189} \leq n
    \end{aligned}
    \end{equation}
    where we used $\alpha= \frac{16}{21}$, $T = T_1 + \tau_i$, and $T_1\geq \frac{n}{3}$.
\end{itemize}

\paragraph{$\bullet$} If type $i$ gets inactive at some point, $\{B_j\}^{T+k}_{j=1}$ is the set of bags assigned before type $i$ gets inactive. Let $\{B_j\}^{n}_{j=T+k+1}$ be the set of bags assigned after type $i$ gets inactive. The latter set of bags can be assigned either during the first or second part of the algorithm. Specifically, when type $i$ becomes inactive, if another active type still exists, the first while loop might continue. As a result, some bags may be assigned while type $i$ is inactive, even though the algorithm is still in the first while loop. We have already shown in \eqref{eq:simple} $v_i(\cup_{j=1}^{T+k}B_j) \leq T\alpha + k$.

From this point onward, let us consider the state of the problem at the moment when type \( i \) becomes inactive. Consider the sets \( M \) and \( H \), representing the available middle-valued and high-valued items, respectively, immediately after type \( i \) becomes inactive. Let \( g = H[-1] \) denote the least valued item in \( H \), and \( g' = M[-1] \) the least valued item in \( M \). Since all bags saved in $F$ while type $i$ was active, and all $T+k$ bags assigned before type $i$ becomes inactive, each contain exactly one HV item, the number of remaining HV items is given by \( |H| = n - T - k - b \).

Let \( t = n - g \) represent the number of unavailable HV items valued no more than $g$. Since \( b \) high-valued items ranked lower than \( g \) have already been saved in some bags in \( F \), and \( \tau_i + k \) high-valued items ranked lower than \( g \) have been assigned, it follows that \( \tau_i + k + b \leq t \leq T_1 + \tau_i + k + b \).

Let \( r = g' - n \) represent the number of unavailable MV items valued no less than $g'$, plus one. Since \( b \) middle-valued items ranked above \( g' \) have already been saved in some bags of \( F \), it follows that \( b \leq r \leq n \). Finally, let \( t' = t - b \), which implies \( \tau_i + k \leq t' \leq T + k \).

Note that the least valued item in \( M \) is positioned at \( r + n \). Since \( b \) middle-valued items ranked above \( M[1] \) have already been saved in some bags in \( F \), the items in \( M \) must belong to the range \( MV[b+1:r] \).

\begin{example}
    For clarity, consider type~2 and the moment it became inactive in \cref{example1}. At that point, the crossed items and the items contained in the red and green bundles in \cref{fig:example2} are unavailable. Thus, for type~2, the least-valued available high-valued (HV) item is item $3$, so we have $g = 3$. Similarly, the highest-valued available middle-valued (MV) item is item $13$, hence $g' = 13$. Given that $n = 9$, we have $t = n - g = 9 - 3 = 6$ and $r = g' - n = 13 - 9 = 4$. \qed
\end{example}

\begin{figure}
    \centering
    \begin{tikzpicture}[every node/.style={draw, circle, minimum size=1cm, inner sep=0pt}]
  \foreach \i in {1,...,9} {
    \ifnum\i=3
    \node[fill=pink] (b\i) at ({(\i-1)*1.5},0) {$g=3$};
  \else
    \node (b\i) at ({(\i-1)*1.5},0) {\i};
  \fi
  }
  
  \foreach \j in {1,...,9} {
    \pgfmathtruncatemacro{\num}{19-\j}
  \ifnum\num=13
    \node[fill=pink] (t\j) at ({(\j-1)*1.5},2) {$g' = 13$};
  \else
    \node (t\j) at ({(\j-1)*1.5},2) {\num};
  \fi  }
  
  \foreach \i in {2,4,9} {
    \draw[line width=1pt] (b\i.north east) -- (b\i.south west);
    \draw[line width=1pt] (b\i.north west) -- (b\i.south east);
  }
    \draw[<->] (4, -1) -- (12.5, -1) node[draw=none, midway, below] {$t=6$}; 

    \draw[<->] (7, 3) -- (12.5, 3) node[draw=none, midway, above] {$r=4$}; 

    \draw[red, thick] (11.3,1) ellipse [x radius=2.2, y radius=0.65, rotate=53, xshift =0pt];

    \draw[red, thick] (9.8,1) ellipse [x radius=2.2, y radius=0.65, rotate=53, xshift =0pt];

    \draw[green, thick] (8.3,1) ellipse [x radius=2.2, y radius=0.65, rotate=53, xshift =0pt];

    \draw[green, thick] (6,1) ellipse [x radius=2.2, y radius=0.65, rotate = 90, xshift =0pt];

    \node[draw=none] at (-1.7, 0) {\textbf{HV :}};

    \node[draw=none] at (-1.7, 2) {\textbf{MV :}};

\end{tikzpicture}
    \caption{This figure represents the set of available HV and MV items when type $2$ gets inactive in \cref{example1}. Note that all crossed items, and items within red and green bundles are unavailable. 
     }
    \label{fig:example2}
\end{figure}
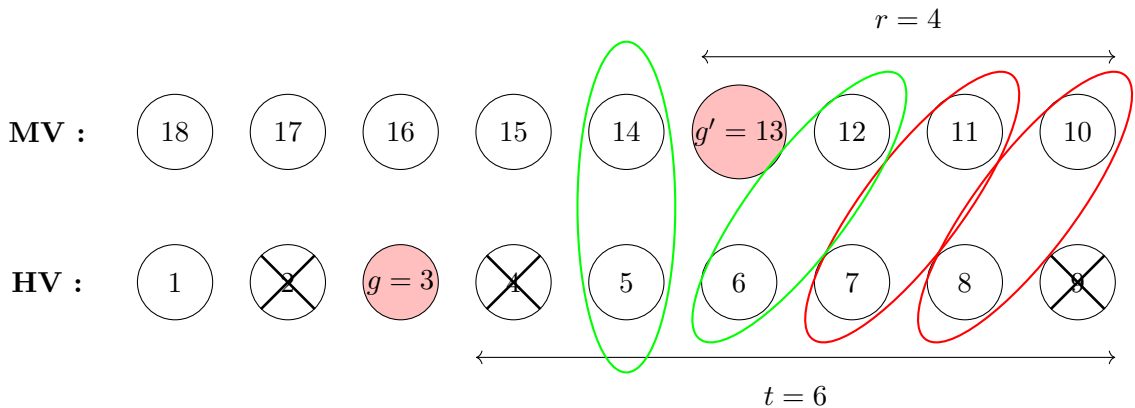

For any assigned bag \( B_j \) after type \( i \) becomes inactive, if \( B_j \) was assigned during the first phase of the algorithm, it contains one high-valued item and one middle-valued item. We consider the middle-valued item in \( B_j \) as the last item added. Prior to adding this item, \( B_j \) contained only a high-valued item, whose value was less than \( \alpha \) for all unsaturated types, as established by Corollary \ref{col:bound}. 

If \( B_j \) was assigned during the second phase of the algorithm, then immediately before the last item was added, the bag's value was less than \( \alpha \) for all unsaturated types, including type \( i \). In both cases, the value of the bag was less than \( \alpha \) before adding the final item. If the value of the last added item is bounded above by \( X \), then the total value of \( B_j \) is upper-bounded by \( \alpha + X \). 

The largest possible upper bound for a bag occurs when the last added item is a middle-valued item. Consequently, the maximum upper bound for \( \cup_{j=T+k+1}^{n} B_j \) is achieved when the number of remaining middle-valued items, \( |M| \), is maximized, and each item in \( M \) is added last to a bag during the bag-filling process. Thus, we proceed under this assumption, considering all \( r-b \) middle-valued items in \( MV[b+1:r] \) to be present in \( M \).

Next, we analyze two possible cases based on whether \( r \) is less than or greater than \( t \). Each case is considered separately, demonstrating that the total value of all assigned bags cannot exceed \( n \). Consequently, the supply of goods remains sufficient throughout the bag-filling process, ensuring that type \( i \) is saturated.

\subsection{Case 1: $r\leq t$}\label{sec:case1}
In this case, we apply the upper bound provided in Corollary \ref{col:bound} to each of the \( r-b \) middle-valued items in \( M \). Consequently, when added to a bag as the final item, each middle-valued item can increase the bag's value to at most \( \frac{3\alpha}{2} \). Given that \( r \leq t \), it follows that \( r-b \leq t-b = t' \), and since \( t' \leq T+k \), we have \( r-b \leq T+k \). Therefore, at most $T+k$ middle-valued items are available. 

Since only \( n-T-k \) bags are required after type \( i \) becomes inactive, the worst-case scenario would involve up to \( \min(n-T-k, T+k) \) bags reaching a value of \( \frac{3\alpha}{2} \). We now examine the following two subcases:

\subsubsection{$T+k \geq \frac{n}{2}$}\label{subsub:tplusk}
In this case $\min(n-T-k, T+k) = n-T-k$, therefore the value of each remained bag can be at most $\frac{3\alpha}{2}$, leading to the following inequality, $v_i(\cup^n_{j=T+k+1} B_j) \leq (n-T-k)\frac{3\alpha}{2}$. By combining this with \eqref{eq:simple}, we derive 
\begin{equation}
    \begin{aligned}
        v_i(\cup^n_{j=1} B_j) & \leq T\alpha + k + (n-T-k)\frac{3\alpha}{2} \\
        & = n\frac{3\alpha}{2} - T\frac{\alpha}{2} + k(1-\frac{3\alpha}{2}) \\
        & \leq n\frac{3\alpha}{2} - \frac{n}{3}\frac{\alpha}{2} + \frac{n}{6}(1-\frac{3\alpha}{2}) \\
        & = n\frac{125}{126} \leq n
    \end{aligned}
\end{equation}
since $T+k \geq \frac{n}{2}$, and $-\frac{\alpha}{2} < 1-\frac{3\alpha}{2}$ when $\alpha=\frac{16}{21}$, the upper bound is obtained when $T$ is the lowest possible number. Since $T\geq \frac{n}{3}$, and $T+k\geq \frac{n}{2}$, the upper bound is obtained when $T=\frac{n}{3}$, and $k = \frac{n}{6}$.  
\subsubsection{$T+k < \frac{n}{2}$}\label{subsub:less}
There are at most $T+k$ middle-valued items. When these items are added as the final items to \( T+k \) bags, the value of each such bag can reach up to \( \frac{3\alpha}{2} \). The remaining \( n-2(T+k) \) bags can have a maximum value of \( \frac{4\alpha}{3} \), as their last added items can only be low-valued, with values upper-bounded by \( \frac{\alpha}{3} \) according to Corollary \ref{col:bound}. Thus, $v_i(\cup^n_{j=T+k+1} B_j) \leq (T+k)\frac{3\alpha}{2} + (n-2(T+k))\frac{4\alpha}{3}$. By combining this with \eqref{eq:simple}, we derive 

\begin{equation}
    \begin{aligned}
        v_i(\cup^n_{j=1} B_j) & \leq T\alpha + k + (T+k)\frac{3\alpha}{2} + (n-2(T+k))\frac{4\alpha}{3} \\
        & = n\frac{4\alpha}{3}-T\frac{\alpha}{6} + k(1-\frac{7\alpha}{6}) \\
        & \leq n\frac{4\alpha}{3}-\frac{n}{3}\frac{\alpha}{6} + \frac{n}{6}(1-\frac{7\alpha}{6}) \\
        & = n\frac{1125}{1134} \leq n. 
    \end{aligned}
\end{equation}
Given $T\geq \frac{n}{3}$ and $T+k \leq \frac{n}{2}$, the upper bound is obtained when $T$ is smallest and $k$ is largest, $T=\frac{n}{3}$ and $k = \frac{n}{6}$ respectively. 

\subsection{Case 2: $r > t$}
When type \( i \) becomes inactive, either \( v_i(H[-1], M[j]) > 1 \) or \( v_i(H[-1], M[j]) < \alpha \) for all \( j \in [|M|] \). Consequently, the value \( v_i(g, g') \) is either greater than 1 or less than \( \alpha \). Therefore, we analyze these two cases separately. 

Let \( l = r - t - 1 \). As $r>t$, \( l \geq 0 \). By the definitions of \( r \) and \( t \), it follows that:\[l = g' - [2n - g + 1].\]

\subsubsection{$v_i(g,g') > 1$}\label{subsub:more} 
The number of middle-valued items can be at most \( r-b \), where \( r-b = l+1+t' \). As discussed earlier, we assume that all \( r-b \) middle-valued items in $MV[b+1:r]$ are present to obtain the greatest upper bound. Given that \( g' = l + 2n - g + 1 \), it follows that \( v_i(2n - g + 1) \geq v_i(g') \). Since \( v_i(g, g') > 1 \), we have \( v_i(g, 2n - g + 1) > 1 \), and by Lemma \ref{lem:upper-13}, it holds that \( v_i(2n - g + 1) < \frac{1}{3} \). 

For all \( j \in [2n - g + 1, g'] \), we have that \( v_i(j) \leq v_i(2n - g + 1) \), and thus \( v_i(j) < \frac{1}{3} \). Consequently, the value of each $l+1$ middle-valued items in \( MV[b+1:r] \) is bounded by \( \frac{1}{3} \). If these items are added as the last item to a bag, the bag's value can increase to at most \( \alpha + \frac{1}{3} \).

On the other hand, for the remaining \( t' \) middle-valued items in \( MV[b+1:r] \) we use the upper-bound provided by Corollary \ref{col:bound}, implying they have a value of at most \( \frac{\alpha}{2} \). When added to a bag, they can increase its value to at most \( \frac{3\alpha}{2} \). Since \( t' \leq T+k \), the number of bags with value at most \( \frac{3\alpha}{2} \) is at most \( \min(T+k, n-T-k) \).

If \( T+k \geq n-T-k \), the analysis in Case \ref{subsub:tplusk} applies. Otherwise, if \( T+k < n-T-k \), at most \( T+k \) bags can have a value of at most \( \frac{3\alpha}{2} \), while the remaining \( n-2(T+k) \) bags have a value upper-bounded by \( \alpha + \frac{1}{3} \). Therefore $v_i(\cup^n_{j=T+k+1} B_j) \leq (T+k)\frac{3\alpha}{2} + (n-2(T+k))\cdot(\alpha+\frac{1}{3})$. By combining this with \eqref{eq:simple}, we derive 
\begin{equation}
    \begin{aligned}
        v_i(\cup^n_{j=1} B_j) & \leq T\alpha + k + (T+k)\frac{3\alpha}{2} + (n-2(T+k))\cdot(\alpha+\frac{1}{3}) \\
        & = n(\alpha+\frac{1}{3})-T(\frac{2}{3}-\frac{\alpha}{2}) - k(\frac{\alpha}{2} - \frac{1}{3}) \\
        & \leq n(\alpha+\frac{1}{3})-T(\frac{2}{3}-\frac{\alpha}{2}) \\
        & \leq n(\alpha+\frac{1}{3})-\frac{n}{3}(\frac{2}{3}-\frac{\alpha}{2}) = n
    \end{aligned}
\end{equation}
where $\alpha = \frac{16}{21}$. Given $T\geq \frac{n}{3}$ and $T+k \leq \frac{n}{2}$, the upper bound is obtained when $T$ and $k$ are smallest, $T=\frac{n}{3}$ and $k = 0$ respectively. 

\subsubsection{$v_i(g,g') < \alpha$}
Since type \( i \) is inactive, it follows that \( v_i(M[j], H[-1]) < \alpha \) or \( v_i(M[j], H[-1]) > 1 \) for all \( j \in [|M|] \), with the additional condition that \( v_i(M[1], H[-1]) > 1 \). Let \( j \) be the largest index such that \( v_i(M[j], H[-1]) > 1 \). By Corollary \ref{col:bound}, we have \( v_i(M[j]) < \frac{\alpha}{2} \), which implies \( v_i(H[-1]) > 1 - \frac{\alpha}{2} \).

Furthermore, since \( v_i(M[j+1], H[-1]) < \alpha \), it follows that \( v_i(M[j+1]) < \frac{3\alpha}{2} - 1 \). Substituting \( \alpha = \frac{16}{21} \), this yields \( v_i(M[j+1]) < \frac{\alpha}{3} \). Consequently, all middle-valued items indexed after \( M[j] \) have a value of at most \( \frac{\alpha}{3} \), which corresponds precisely to the upper bound for low-valued items. 

Thus, based on the position of \( M[j] \), the problem can be reduced to one of the previously analyzed cases. Let \( g' = M[j] \) and define \( r = g' - n \). If \( r \leq t \), the analysis in Case \ref{sec:case1} applies; otherwise, if \( r > t \), the analysis in subcase \ref{subsub:more} applies.

\end{document}